\documentclass[
11pt,
paper=letter,
DIV=18,
oneside,
headings=small,
]{scrartcl}
\usepackage{
amsmath,
amsfonts,%
amssymb,%
color,
enumerate,
hyperref,
framed
} 

\usepackage{privmath,privref, privalgorithm}
\usepackage[disable,textsize=tiny]{todonotes}
\usepackage[short]{pwbib}
\usepackage{pslatex}
\usepackage{paralist}

\SetKwFunction{rank}{rank}
\SetKwFunction{mod}{mod}
\SetKwFunction{loose}{loose}
\SetKwFunction{win}{win}
\SetKwFunction{getName}{getName}
\SetKwFunction{newScan}{newScan}
\SetKwFunction{Update}{Update}
\SetKwFunction{relName}{relName}
\SetKwFunction{collect}{collect}
\SetKwFunction{array}{array}
\SetKwFunction{True}{True}
\SetKwFunction{False}{False}

\def\Reg{\mathcal{R}}
\def\SmallRegs{\mathcal{T}}
\def\Proc{\mathcal{P}}
\def\Wrong{\mathcal{F}}
\def\Writer{\mathcal{O}}

\newcommand{\complete}[1]{\FuncSty{$#1$}-complete}
\newcommand{\val}[2]{\nu_{#1}(#2)}
\newcommand{\fullset}[1]{\FuncSty{$#1$}-full-set}
\newcommand{\Procs}[1]{\mathrm{Procs}{(\FuncSty{$#1$})}}
\newcommand{\Index}[1]{\mathrm{index}{(\FuncSty{$#1$})}}
\newcommand{\Names}[1]{\mathrm{Names}{(\FuncSty{$#1$})}}
\newcommand{\name}[2]{\text{name}($#1,#2$)}
\newcommand{\writer}[2]{\text{writer}\FuncSty{{$(#1$$,$$#2)$}}}
\newcommand{\writeOp}[1]{\mathrm{writeOp}\FuncSty{{($#1$)}}}
\newcommand{\writers}[1]{\mathcal{Z}\FuncSty{{($#1$)}}}

\usepackage{scrpage2}
\clearscrheadings
\clearscrplain
\cofoot{\pagemark}
\advance\topmargin-.5cm
\advance\textheight1cm
\advance\footskip-.5cm

\usepackage{amsthm}
\newtheorem{theorem}{Theorem}[section]
\newtheorem{lemma}[theorem]{Lemma}
\newtheorem{observation}[theorem]{Observation}

\newtheorem{corollary}{Corollary}[section]

\newtheoremstyle{upshape}{}{}{\upshape}{}{\bfseries}{.}{ }{}
\theoremstyle{upshape}

\title{Space Bounds for Adaptive Renaming}
\author{Maryam Helmi, Lisa Higham, Philipp Woelfel}

\begin{document}
\begin{titlepage}
  \begin{center}\sffamily
  \LARGE
  {\bfseries Space Bounds for Adaptive Renaming} \\[\medskipamount]
  
  \vskip\bigskipamount

  Maryam Helmi\footnote{\texttt{mhelmikh@ucalgary.ca}, +1\,403\,210-9416},
  Lisa Higham\footnote{\texttt{higham@ucalgary.ca}, +1\,403\,220-7696},
  and Philipp Woelfel\footnote{\texttt{woelfel@ucalgary.ca}, +1\,403\,220-7259}\\[\medskipamount]
  
  \large University of Calgary\\ 
  Department of Computer Science\\
  Calgary, T2N1N4 Alberta, Canada
  
  \end{center}
\vskip\bigskipamount

\begin{abstract}
  \noindent{\bfseries Abstract.}
  We study the space complexity of implementing long-lived and one-shot adaptive renaming from multi-reader multi-writer registers, 
in an asynchronous distributed system with $n$ processes. 
As a result of an \emph{$f$-adaptive renaming algorithm} each participating process gets a distinct name in the range $\{1,\dots,f(k)\}$ provided $k$ processes participate.

Let $f: \{1,\dots,n\} \rightarrow \mathbb{N}$ be a non-decreasing function satisfying $f(1) \leq n-1$ and let $d = \max\{x ~|~ f(x) \leq n-1\}$. 
We show that any non-deterministic solo-terminating long-lived $f$-adaptive renaming object requires $d + 1$ registers. 
This implies a lower bound of $n-c$ registers for long-lived $(k+c)$-adaptive renaming, which we observe is tight. 

We also prove a lower bound of $\lfloor \frac{2(n - c)}{c+2} \rfloor$ registers for implementing any non-deterministic solo-terminating 
one-shot $(k+c)$-adaptive renaming. 
We provide two one-shot renaming algorithms: a wait-free algorithm and an obstruction-free algorithm. 
Each algorithm employs a parameter to depict the tradeoff between space and adaptivity. When these parameters are chosen appropriately, this results in 
a wait-free one-shot $(\frac{3k^2}{2})$-adaptive renaming algorithm from $\lceil \sqrt{n} \rceil + 1$ registers,  
and an obstruction-free one-shot $f$-adaptive renaming algorithm from
only $\min\{n, x ~|~ f(x) \geq 2n\} + 1$ registers.
\end{abstract}
\end{titlepage}

\newpage
\setcounter{page}{1}
\section{Introduction}

Distributed systems with a large number of processes, such as the Internet, provide services that 
are typically used by only a small number of processes simultaneously. 
This is problematic if the time or space used by the service is a function of the size 
of the name-space of the processes that could use it. 
The time or space consumed by such applications can be significantly decreased by having each process that wants to
use the application first acquire a temporary name from a name space that is adequate to distinguish all the participants,
but much smaller than the name-space of the distributed system, and then return the temporary name to the pool 
when it is finished with the service.
This is the role of a shared renaming object.
A related application of the renaming object is in operating systems where processes repeatedly acquire and release 
names that correspond to a limited number of resources~\cite{BP89a}.
Renaming is an important tool in distributed computing~\cite{ABDPR1990a} 
because it enhances the practicality and usefulness of distributed system services. 
A renaming object may be even more useful if  the time and space resources it consumes is 
a reasonable function of the actual number of processes that are currently either holding, acquiring, or releasing a name. In this paper, we address the renaming problem for the standard asynchronous shared memory model with $n$ processes.  

With an $f$-adaptive renaming object, each of the $n$ processes can perform a \getName{} 
and return a distinct name in a small domain $\{1,\dots,f(k)\}$ 
where $k$ is the number of participants. Herlihy and Shavit~\cite{HS99a}, and also Rajbaum and Casta\~{n}eda~\cite{CR2008} 
showed that there is no deterministic, wait-free implementation of $(2k-2)$-adaptive 
renaming from multi-reader multi-writer registers.
This result also follows from the relationship between the adaptive renaming problem 
and \emph{strong symmetry breaking} (SSB):
a $(2k-2)$-renaming algorithm implies a solution to SSB~\cite{AP12},
for which there is no deterministic wait-free solution~\cite{AP12,BG93,G06}. 
This impossibility can be circumvented by using randomness or stronger primitives 
such as compare-and-swaps~\cite{AACSZ2011,AAGGG2010a,EHW98,MoirG96,PPTV98}. 
The step complexity of deterministic and randomized algorithms has been studied extensively 
in asynchronous systems (see e.g., \cite{AFWR99,AAGG11,AAGGG2010a,BEW2011a,EW2013a}). 
However, there are no previous results on the space complexity of adaptive renaming. 
Because renaming seems to require that participants discover information about each other,
adaptive renaming appears related to $f$-adaptive collect. 
A collect algorithm is $f$-adaptive to total contention, if its step complexity is $f(k)$, where $k$ is the number of processes that ever took a step. 
Attiya, Fich and Kaplan~\cite{AFK2004}, proved that $\Omega(f^{-1}(n))$ multi-reader multi-writer registers are 
required for $f$-adaptive collect. 

Suppose you have $m$ shared registers available to construct a renaming object for a system with $n$ processes.  
First we would like to know under what additional conditions such an implementation exists, and when it does, 
how best to use the $m$ registers.
Suppose, when there are $k$ participants, the acquired names are in the range $\{1, \dots, f(k)\}$. 
Will $f(k) = k ^c$ for a small constant $c$ suffice for the application? 
Must $f(k)$ be closer to $k$, say within a constant? Perhaps it should even be exactly $k$ (\emph{tight} adaptive renaming)?
Does the application need to permit processes to repeatedly acquire and release a name (\emph{long-lived} renaming),
or do processes get a name at most once (\emph{one-shot} renaming)? 
How strong a progress guarantee is required? 
Is the number of participants usually less than some bound $b$ much smaller than $n$?
If so, is there still some significant likelihood that the number of participants is somewhat bigger than $b$, 
or is there confidence that the bound $b$ is never, or only very rarely, exceeded?
In the rare cases when there are a large number of participants, 
can the system tolerate name assignments from a very large name space?

In order to study the space complexity implication for these questions,
we first generalize the adaptive renaming definition.
Both versions (\emph{long-lived} and \emph{one-shot}) of \emph{$b$-bounded $f$-adaptive renaming} support the operation \getName{}, which returns a name to each invoking process. The long-lived version also supports an operation \relName{}, which releases the name to the available name domain. 
Both versions must satisfy 1) no two processes that have completed a \getName{} and have not started their following \relName{}, receive the same name, 2) if there are $k \leq b$ processes that have invoked \getName{} and have not completed their subsequent \relName{} during an execution of \getName{} by process $p$, then $p$ returns a name in $\{1, \dots , f(k)\}$.  
Observe that for the one-shot case, $k$ is the number of processes that have started a \getName{} before $p$ completes its \getName{}. 
We call the problem of $n$-bounded $f$-adaptive renaming simply $f$-adaptive renaming. 
The special case when $f(k) = k$ and $b = n$ is called \emph{tight} renaming.
Our goal is to determine the relationships between $b$, $f(k)$, and $m$ for one-shot versus long-lived, 
and wait-free versus non-deterministic solo-terminating implementations of adaptive renaming objects from multi-reader/multi-writer registers.

Let $f: \{1,\dots,n\} \rightarrow \mathbb{N}$ be a non-decreasing function satisfying $f(1) \leq n-1$ and let $d = \max\{x ~|~ f(x) \leq n-1\}$. Note that if $f(1) \geq n$, $f$-adaptive renaming is a trivial problem.  
In this paper we show:
\begin{compactitem}
\item
At least $d + 1$ registers are required to implement any 
non-deterministic solo-terminating long-lived $d$-bounded $f$-adaptive renaming object. 
\item
At least $\lfloor \frac{2(n - c)}{c+2} \rfloor$ registers are required to implement 
any non-deterministic solo-terminating one-shot $(k+c)$-adaptive renaming object where, $c$ is any non-negative integer constant. 
\item
For any $b \leq n$, 
there is a wait-free one-shot $(b-1)$-bounded $(k(k+1)/2)$-adaptive renaming algorithm implemented from $b$ bounded registers.
When $k \geq b$, the returned names are in the range $\{1,\ldots,n+\frac{b(b-1)}{2}\}$.
\item
For any $b \leq n$, 
there is an obstruction-free one-shot $(b-1)$-bounded $k$-adaptive renaming algorithm implemented from $b+1$ 
bounded registers.
When $k \geq b$, the returned names are in the range $\{1,\ldots,n+b-1\}$.
\end{compactitem}\

By using these results and setting $b$ appropriately we then derive the following corollaries: 
\begin{compactitem}
	\item A wait-free one-shot $(\frac{3k^2}{2})$-adaptive renaming algorithm 
that uses only $\lceil \sqrt{n} \rceil + 1$ registers.
	\item An obstruction-free one-shot $f$-adaptive renaming algorithm 
that uses only $\min\{n, x ~|~ f(x) \geq 2n\} + 1$ registers.
	\item A tight space lower bound of $n-c$ registers for long-lived $(k+c)$-adaptive renaming for any integer constant $c \geq 0$.
\end{compactitem}


Our lower bound proofs use covering techniques first introduced by Burns and Lynch \cite{BurnsL93}.  
The main challenge is 
to exploit the semantics of the renaming object 
to force the processes to write to a large number of registers.
In the lower bound for the one-shot case, 
we first build an execution in which some processes are poised to write to (\emph{cover}) a set of registers.
Then we argue that if enough new processes take steps after this, 
at least one of them must become poised to write to a register not already covered, 
since, otherwise, the covering processes can obliterate all the traces of the new processes, 
causing some \getName{} to return an incorrect result. 
For the lower bound for the long-lived case, we exploit that fact that processes can perform \getName{} and \relName{} 
repeatedly to build a long execution, where 
in each inductive step either another register is covered or an available name is used up 
without being detected by other processes. 
 \section{Preliminaries}

This section describes our model of computation and the notation, vocabulary and general techniques used in this paper.
Previous work by many researchers (for example \cite{AW98,BurnsL93,FHS1998a,FR03,HHPW14,L96}) 
has collectively developed similar tools that serve to make our description of results and presentation of proofs precise, concise and clear. 
Much of the terminology presented in this section is borrowed or adapted from this previous research. 

Our computational model is an asynchronous shared memory system consisting of $n$ processes $\Proc =  \{p_1,\dots,p_n\}$ 
and $m$ shared registers $\Reg = \{ R_1,\dots,R_m \}$. 
Each process executes code that can access its own independent random number generator 
and its own private registers as well as the shared registers. 
Each shared register supports two operations, read and write. 
Each such operation happens atomically in memory.
Processes can only communicate via those operations on shared registers.
The algorithm is \emph{deterministic} if each process' code is deterministic; 
that is, no process' code contains any random choice.  

Informally, an execution arises one step at a time, 
where a step consists of some process, chosen arbitrarily,
executing the next instruction in its code.
This instruction could be a shared memory access, or a local memory access, 
or a local operation including, possibly, a random choice.
Notice, however, that after a process takes a shared memory step, 
the outcome of all its subsequent local operations and
random choices up to (but not including) its next shared memory operation 
is independent of any intervening operations by other processes.  
Therefore, there is no loss of generality in assuming that a step by a process consists of 
a single shared memory access (or, initially, its method-call invocation)
followed by all its subsequent local operations and random choices, 
up to the point where that process is poised to execute its next shared memory operation.

A \emph{configuration} $C$ is a tuple $(s_1,\dots,s_n,v_1,\dots,v_m)$, denoting that 
process $p_i$, $1\leq i\leq n$, is in state $s_i$ and
register $r_j$, $1\leq j\leq m$, has value $v_j$.
Furthermore the state $s_i$ of $p_i$ is one in which 
$p_i$'s next operation is either a shared memory operation or an invocation of 
a method-call (\getName{} or \relName{}).
Configurations will be denoted by capital letters. 
The initial configuration, where each process' next step is to invoke a method-call, 
is denoted $C^\ast$.

Given a configuration, $C$, a \emph{step from $C$} 
is a pair of the form $(p, \tau)$ where $p$ is a process identifier, 
and $\tau$ is a sequence of outcomes that arise from the sequence of all random choices that $p$ makes after completing its pending shared memory operation starting from configuration $C$ up to the point where $p$ is poised to do its next shared memory operation.
An execution is an alternating sequence of configurations and steps starting and ending with a configuration, and defined inductively as follows.
The 0-step or empty execution starting at $C$ is $(C)$.
A $k$-step execution, $k \geq 1$,  starting at $C_0$ is a sequence 
$(C_0, (q_1,\tau_1), C_1, (q_2, \tau_2), \ldots , (q_k, \tau_k), C_k)$ where 
\begin{compactitem}
	\item $(C_0, (q_1,\tau_1), C_1, (q_2, \tau_2), \ldots , (q_{k-1}, \tau_{k-1}), C_{k-1})$ is a $k-1$ step execution starting at $C_0$, and
	\item $(q_k, \tau_k)$ is a step from $C_{k-1}$ and $C_k$ is the configuration resulting from that step.
\end{compactitem}
An \emph{execution} is a $k$-step execution for any integer $k \geq 0$.
A subsequence, $\sigma = ((q_1, \tau_1), \ldots, (q_k,\tau_k))$, consisting of the steps from an execution starting at $C$ is called a \emph{schedule starting at $C$}.
If $\sigma$ is a schedule starting at $C$, then the execution starting at $C$ arising from $\sigma$ is
abbreviated $E= (C;\sigma)$ and $\sigma(C)$ denotes the final configuration of $E$.  
If an algorithm is deterministic, then the second component of every step of every execution of the algorithm is empty because there are no random choices.
So in this case a schedule is simplified to just a sequence of process identifiers.

A configuration, $C$, is \emph{reachable} if there exists a finite schedule, $\sigma$, such that $\sigma(C^\ast)= C$. 
Let $\sigma$ and $\pi$ be two finite schedules such that $\sigma$ starts at configuration $C$ and $\pi$ starts at $\sigma(C)$. 
Then $\sigma\pi$ denotes the concatenation of $\sigma$ and $\pi$, and is a schedule starting at $C$. 
Let $P \subseteq \Proc$ be a set of processes,  and $\sigma$ a schedule. 
We say $\sigma$ is \emph{$P$-only} if all the identifiers of processes that appear in $\sigma$
are in $P$. 
If the set $P$ contains only one process, $p$, then we say $\sigma$ is \emph{$p$-only}.
We denote the set of processes that appear in schedule $\sigma$ by $\mathrm{procs}(\sigma)$. 

A deterministic implementation of a method is \emph{wait-free} if, 
for any reachable configuration $C$ and any process $p$, 
$p$ completes its method call in a finite number of its own steps, 
regardless of the steps taken by other processes. 
An implementation of a method is 
\emph{non-deterministic solo-terminating} if,
for any reachable configuration $C$ and any process $p$, 
there exists a finite $p$-only schedule, $\sigma$, starting from $C$ 
such that $p$ has finished its method call in configuration $\sigma(C)$\cite{FHS1998a}.
Non-deterministic solo-termination for deterministic implementations is called \emph{obstruction-free}.



We say process $p$ \emph{covers} register $r$ in a configuration $C$, 
if $p$ writes to $r$ in its next step. 
A set of processes $P$ covers a set of registers $R$ if for every register $r\in R$ there is a process $p\in P$ 
such that $p$ covers $r$.
If $|P| = |R|$, 
then we say $P$ \emph{exactly} covers $R$. 
Consider a process set $P$ that exactly covers the register set $R$ in configuration $C$. 
Let $\pi_P$ be any permutation which includes exactly one step by each process in $P$. 
Then the execution $(C;\pi_{P})$ is called a \emph{block-write} by $P$ to $R$. 
Two configurations 
$C=(s_1,\ldots,s_n,v_1,\ldots,v_m)$ and 
$C'=(s'_1,\ldots,s'_n,v'_1,\ldots,v'_m)$ are \emph{indistinguishable} to process $p_i$ 
if $s_i=s'_i$ and $v_j=v'_j$ for $1\leq j\leq n$. 
Let $P$ be a set of processes, and $\sigma$ any $P$-only schedule starting at configuration $C$.
If for every process $p\in P$, $C$ and $C'$ are indistinguishable to $p$, 
then  $\sigma$ is also a schedule starting at $C'$ and $\sigma(C)$ and $\sigma(C')$ are indistinguishable to $p$.

A process $p$ \emph{participates in configuration $C$} 
if in $C$, $p$ has started a \getName{} operation and has not completed the following \relName{}. 
A process is called \emph{idle in configuration $C$} if it does not participate in $C$. 
A configuration $C$ is called \emph{quiescent} if, $\forall p \in \Proc$, $p$ is idle in $C$.
We say process $p$ \emph{owns name $x$ in configuration $C$} 
if in $C$, $p$ has completed a \getName{} operation that returned name $x$ and $p$ has not started \relName{}. 
Let $C_0, \dots, C_e$ be a sequence of configurations arising from execution $E$. 
The number of \emph{participants in $E$} is the maximum over all $i$, $0 \leq i \leq e$, of the number of participants in $C_i$. 
Given these definitions, the definition of a renaming object can be stated more precisely as follows. 
Let $f: \{1,\dots,n\} \rightarrow \mathbb{N}$ be a non-decreasing function satisfying $f(1) \leq n-1$. 
Both long-lived and one-shot $b$-bounded $f$-adaptive renaming support the operation \getName{}. Operation \getName{} by process $p$ returns a name $x$ to $p$. 
The long-lived version also supports the operation \relName{}, which releases the name $x$. 
Both versions must satisfy
1) there is no reachable configuration in which two processes own the same name,  
2) if the number of participants during $p$'s \getName{}, $k$, is at most $b$ then, $x \in \{1, \dots , f(k)\}$. Observe that, properties 1) and 2) imply that $f(k) \geq k$ for all $k \in \{1, \ldots, n\}$. 
\section{A Space Lower Bound for Long-Lived Loose Renaming Objects}\label{sec:long-lived-lwb}

For any non-decreasing function $f$ satisfying $f(1) \leq n-1$, let $d$ be the largest integer such that $f(d) \leq n-1$.
We prove that at least $d + 1$ registers are required for 
non-deterministic solo-terminating long-lived $f$-adaptive renaming in our system. 
The proof relies on two lemmas. 
Lemma~\ref{lemma:long-lived-help} says that there is no reachable configuration $C$ in which $n-d$ 
processes own names in the range $\{1,\dots,n-1\}$ 
while all of the other $d$ processes are idle and unaware of any of the processes with names. 
The intuition for this proof is simple: 
if such a reachable configuration $C$ exists,
then there is a configuration reachable from $C$ in which $(n-d) + d = n$ processes 
all own names in the range $\{1,\dots,n-1\}$. 
Lemma~\ref{lemma:long-lived-induction} provides the core of the lower bound argument and the intuition is as follows.   
Let $C$ be any reachable configuration in which fewer than $n-d$ processes 
own names in the range ${\{1,\dots,n-1\}}$ 
while all of the other $d + 1$ processes are idle and unaware of the processes with names.
Then there is a reachable configuration from $C$ in which either $d + 1$ distinct registers are written, or
one more name is owned, and the unnamed processes are again idle and still unaware of the processes with names.  
Since the initial configuration has no processes with names, and all processes are idle, 
we can apply Lemma \ref{lemma:long-lived-induction} repeatedly until either we have exactly covered $d + 1$ registers or we reach
a  configuration in which $n-d-1$ processes own names in the range $\{1,\dots,n-1\}$.
Since, according to Lemma \ref{lemma:long-lived-help}, we cannot get beyond an $(n-d-1)$-invisibly-named configuration, 
we must eventually exactly cover $d + 1$ registers, completing the proof.
We will see, in the formal proof, that the result applies even when the renaming implementation is $(d + 1)$-bounded.
 
The definitions and lemmas that follow refer to any non-deterministic solo-terminating implementation from shared registers 
of a long-lived $f$-adaptive renaming object. 
For a configuration $C$ and a set of processes $Q$, we say $Q$ is \emph{invisible} in $C$, 
if there is a reachable quiescent configuration $D$ such that $C$ and $D$ are indistinguishable to all processes in $\overline{Q}$.
If the set $Q$ contains only one process $q$, then we say process $q$ is invisible. 
Configuration $C$ is called $\ell$-\emph{invisibly-named}, if 
there is a set $Q$ of $\ell$ processes, such that in $C$ every process in $Q$ owns a name in $\{1,\dots,n-1\}$ and $Q$ is invisible.

\begin{lemma}
\label{lemma:long-lived-help}
For the largest integer $d$ satisfying $f(d)\leq n-1$, there is no reachable $(n-\d)$-invisibly-named configuration.
\end{lemma}

\begin{proof}
By way of contradiction, suppose that there exists a set $Q$ of $n-d$ processes such that in configuration $C$, 
all processes in $Q$ are invisible and own names in the range $\{1,\dots,n-1\}$.
Since $Q$ is invisible in $C$, there is a reachable quiescent configuration $D$ such that $D$ and $C$ are indistinguishable to $\overline{Q}$.
Let $\sigma$ be a $\overline{Q}$-only schedule such that in execution $(D; \sigma)$, 
all processes in $\overline{Q}$ perform a complete \getName{}. 
Because $|\overline{Q}|=d$ 
all processes in $\overline{Q}$ get names in the range $\{1,\dots,f(d)\} \subseteq \{1,\dots, n-1 \}$.
Since $C$ and $D$ are indistinguishable to $\overline{Q}$, 
all processes in $\overline{Q}$ perform a complete \getName{} during $(C; \sigma)$ 
and get names in the range $\{1,\dots, f(\d) \}$ as well. 
Therefore in configuration $\sigma(C)$ all processes in $Q \cup \overline{Q}$ have names in the range $\{1,\dots,n-1\}$. 
However $|Q \cup \overline{Q}|=n$.  
This is a contradiction because this implies that the acquired names are not distinct.
\end{proof}

The intuition for Lemma~\ref{lemma:long-lived-induction} is as follows.  
Recall that in an $\ell$-invisibly-named configuration,
$\ell$ processes have names, the $n-\ell$ others are idle and unaware of the presence of the invisibly-named processes, and no register is covered.
Starting from this configuration we select one process at a time from the set of idle processes and let it execute until 
either it covers a  register not already covered, 
or it gets a name without covering a new register.
We continue this construction as long as the selected process covers a new register.
If we reach $d + 1$ processes covering distinct registers we are done. 
Otherwise, we reached a configuration in which one more process holds a name.  
Furthermore, we can obliterate the trace of this process with the appropriate block write, 
and then let all other non-idle processes complete their \getName{} methods and the following \relName{}.
This takes us to an $(\ell+1)$-invisibly-named configuration.

\begin{lemma}\label{lemma:long-lived-induction}\samepage
Let $d$ be the largest integer such that $f(d)\leq n-1$.
For any $0\leq \ell \leq n-d - 1$ and any reachable $\ell$-invisibly-named configuration $C$, 
there exists a schedule $\sigma$, where $|\mathrm{procs}(\sigma)| \leq d + 1$, 
and either
\begin{compactitem}
\item in configuration $\sigma(C)$ at least $d + 1$ distinct registers are exactly covered; or
\item configuration $\sigma(C)$ is $(\ell+1)$-invisibly-named.
\end{compactitem}

\end{lemma}

\begin{proof}

Let $C$ be an $\ell$-invisibly-named configuration, and let $Q$ be the set of $\ell$ processes that are invisible in $C$.  
Let $D$ be a quiescent configuration that is indistinguishable from $C$ for all processes in $\overline{Q}$. 
First, we inductively construct a sequence of schedules $\delta_0,\delta_1,\ldots$ until we have constructed $\delta_{\mathit{last}}$ 
such that in $\delta_{\mathit{last}}(C)$ either 
\begin{compactenum}[a)]
	\item $d + 1$ registers are exactly covered, or,
	\item $(\ell + 1)$ processes own names in $\{1, \dots , n-1\}$.
\end{compactenum}

We maintain the invariant that for each $i\in\{0,\dots,\mathit{last}\}$ in configuration $\delta_i(C)$, 
a set $P_i$ of $i$ processes exactly covers a set $L_i$ of $i$ distinct registers, $P_i\cap Q=\emptyset$, and $\delta_i$ is $P_i$-only. 
Let $\delta_0$ be the empty schedule. Then in configuration $\delta_0(C)=C$, no register is covered, so the invariant is true for $P_0=L_0=\emptyset$.

Now consider $i\geq 0$.
If $a)$ or $b)$ holds for $\delta_i$, we let $\mathit{last}=i$ and are done.
Otherwise, since in $\delta_i(C)$ a set $L_i$ of $i$ distinct registers is covered, we have $i \leq d $.
We construct $\delta_{i+1}$ as follows.
Select $p \in \overline{P_i \cup Q}$. Let $\gamma$ be the shortest $p$-only schedule such that either
\begin{compactenum}[1)]
	\item $p$ does a complete \getName{} in execution $(\delta_i(C); \gamma)$, or
	\item in configuration $\gamma(\delta_i(C))$, $p$ covers a register $r \notin L_i$.
\end{compactenum}

Let $\delta_{i+1}$ be $\delta_i \gamma$.
First assume case $1)$ happens.
By construction the process that performs $\delta_1$ does not write to any register.
If $i = 0$ and $p$ does a complete \getName{} in execution $(\delta_0(C); \gamma)$, then $\mathit{last} = 1$ and we are done. 
For any $i \geq 1$, because $Q$ is invisible to $p$, in $(\delta_i(C);\gamma)$ $p$ becomes aware of at most the $i -1 \leq d $ other processes in $P_{i}$.
Since $f(d)\leq n-1$, $p$ gets a name in $\{1,\dots,n-1\}$, and thus in configuration $\delta_{i+1}(C)$ all processes in $Q \cup \{p\}$ 
own names in $\{1, \dots , n-1 \}$ and $|Q \cup \{p\}|=\ell + 1$.
So condition $b)$ is achieved, the construction stops and $\delta_{\mathit{last}}=\delta_{i+1}$.

Now suppose case $2)$ happens.
If $i+1 =d +1$, then condition $a)$ is achieved, the construction stops and $\delta_{\mathit{last}}=\delta_{i+1}$. 
Otherwise, the invariant remains satisfied for $L_{i+1}=L_{i} \cup \{r\}$ and $P_{i+1}=P_i \cup \{p\}$.
Clearly, after at most $d + 1$ steps either $a)$ or $b)$ is achieved.  

Now, using schedule $\delta_{\mathit{last}}$ we construct  schedule $\sigma$. 
If $\delta_{\mathit{last}}(C)$ satisfies $a)$, let $\sigma=\delta_{\mathit{last}}$ and the lemma holds.
Hence, suppose that $\delta_{\mathit{last}}(C)$ satisfies $b)$.
Let $\alpha$ be the $P_{\mathit{last}-1}$-only schedule such that in execution $(\delta_{\mathit{last}}\pi_{P_{\mathit{last}-1}}(C); \alpha)$ 
every process $q \in P_{\mathit{last}-1}$ completes its pending \getName{} operation and performs a complete \relName{}. 
During execution $(C; \delta_{\mathit{last}})$ only registers in $L_{\mathit{last}-1}$ were written and 
in configuration $\delta_{\mathit{last}}(C)$,  $P_{\mathit{last}-1}$ exactly covers these registers.
Because $\delta_{\mathit{last}} = \delta_{\mathit{last}-1}\gamma$ for some $p$-only postfix $\gamma$ of $\delta_{\mathit{last}}$, after a block write by $P_{\mathit{last}-1}$,
configurations $\delta_{\mathit{last}} \pi_{P_{\mathit{last}-1}} (C)$ and $ \delta_{\mathit{last}-1} \pi_{P_{\mathit{last}-1}} (C)$ 
are indistinguishable to $\overline{Q \cup \{p\}}$.
Since $C$ and $D$ are indistinguishable to $\overline{Q}$, 
configurations $\delta_{\mathit{last}-1} \pi_{P_{\mathit{last}-1}}(C)$ and $\delta_{\mathit{last}-1} \pi_{P_{\mathit{last}-1}}(D)$ 
are also indistinguishable to $\overline{Q}$.
So, 
configurations $\delta_{\mathit{last}} \pi_{P_{\mathit{last}-1}} (C)$ and $ \delta_{\mathit{last}-1} \pi_{P_{\mathit{last}-1}} (D)$ 
are indistinguishable to $\overline{Q \cup \{p\}}$.
Hence, configurations $\delta_{\mathit{last}}\pi_{P_{\mathit{last}-1}}\alpha(C)$ and $\delta_{\mathit{last}-1}\pi_{P_{\mathit{last}-1}}\alpha(D)$ 
are indistinguishable to $\overline{(Q \cup \{p\})}$.
Since $\delta_{\mathit{last}-1}\pi_{P_{\mathit{last}-1}}\alpha(D)$ is quiescent,
configuration $\delta_{\mathit{last}}\pi_{P_{\mathit{last}-1}}\alpha(C)$ is an $(\ell+1)$-invisibly-named configuration. 
Therefore, the lemma holds for $\sigma=\delta_{\mathit{last}}\pi_{P_{\mathit{last}-1}}\alpha$.
\end{proof}

\begin{theorem}\label{theorem:long-lived}\samepage
Let $d$ be the largest integer such that $f(d) \leq n-1$.
Any non-deterministic solo-terminating implementation of a long-lived $d$-bounded $f$-adaptive renaming object 
requires at least $d + 1$ registers.
\end{theorem}
\begin{proof}
  Note that $C^\ast$ is a reachable $0$-invisibly-named configuration.
  We iteratively construct a sequence of schedules $\sigma_0,\sigma_1,\dots,\sigma_{\mathit{last}}$ as follows:
  If $0\leq i\leq n-d$ and $C_i$ is a reachable $i$-invisibly-named configuration, we apply Lemma~\ref{lemma:long-lived-induction} to obtain a schedule $\sigma_i$, $|\mathrm{procs}(\sigma_i)|\leq d + 1$, such that $C_{i+1}=\sigma_i(C_i)$ is either an $(i+1)$-invisibly-named configuration, or in $C_{i+1}$ at least $d + 1$ distinct registers are covered.
  In the latter case we let $\mathit{last}=i+1$ and finish the iterative construction.  
  By Lemma~\ref{lemma:long-lived-help}, there is no $(n-d)$-invisibly-named configuration. Hence if the iterative construction reaches a $(n-d -1)$-invisibly-named configuration, by Lemma \ref{lemma:long-lived-induction}, there is a reachable configuration, in which $d + 1$ registers are covered.
\end{proof}

\begin{corollary}\label{theorem:long-lived-corol}\samepage
  Let $c\in\{1,\dots,n-1\}$ and  $b = n - c$.
Any non-deterministic solo-terminating implementation of a long-lived $b$-bounded $(k + c)$-adaptive renaming object 
requires at least $b$ registers. 
\end{corollary}
\section{A Space Lower Bound for One-shot Additive Loose Renaming}
\label{sec:one-shot-lwb}

In one-shot renaming, each process is constrained to call \getName{} at most once (and does not invoke \relName{}),
which imposes a severe restriction on the techniques available for proving lower bounds.
In particular, constructions that rely on processes repeatedly getting and releasing names cannot be used for one-shot lower bounds. 
We observed however, that a straightforward linear lower bound for tight renaming actually applies even for one-shot adaptive renaming.
Thus, we are motivated to study one-shot renaming objects with looseness constrained by a constant, 
specifically $k$-renaming and $(k + c)$-renaming. 
We refer to one-shot $(k+c)$-renaming object as an additive loose renaming object, where $k$ is the number of participants 
and $c \geq 0$ is an integer constant.
For the case $c=0$, it is called an adaptive tight renaming object.

Our lower bound proof has a recursive structure and it relies on a generalization of additive loose renaming as follows. For any set $T \subset \{1, \dots, k + c\}$ where $|T| \leq c$, a $[(k + c) \backslash T]$-renaming object for $k$ processes requires that each participating process returns a unique name from the range $\{1, \dots, k + c\} \backslash T$.

\begin{lemma}\label{lem:domain-shrink}
Any implementation of $[(k + c) \backslash T]$-renaming uses at least as many registers as an implementation of $[(k + c - |T|) \backslash \emptyset]$-renaming.
\end{lemma}
\begin{proof}
Let $A$ be a $[(k + c) \backslash T]$-renaming algorithm. 
Then we construct $[(k + c - |T|) \backslash \emptyset]$-renaming algorithm $A'$ from $A$ without any additional registers as follows.
If $A$ returns name $x$, then $A'$ returns $x - |\{t \in T~|~ t \leq x\}|$. 
Since $A$ returns distinct names in the range $\{1, \dots, k + c\} \backslash T$, obviously $A'$ returns distinct names in the range $\{1, \dots, k + c - |T|\}$. 
\end{proof}

A process is called \emph{startable in configuration $C$} if in $C$, it has not started a \getName{}. 
Since in one-shot renaming, there is no \relName{} operation, in our proofs in this section instead of using quiescent configurations we are interested in configurations in which each process either has completed its \getName{} operation or it has not started a \getName{}. We call such configurations, \emph{quiet} configurations.

\begin{lemma}\label{lemma:add-help}
  Let $D$ be a reachable quiet configuration with $n' \geq c+2$ startable processes. For every startable process $p$, let $\sigma_{p}$ denote a $p$-only schedule such that $p$ performs a complete \getName{} in execution $(D;\sigma_{p})$.
	Let $Q$ be any subset of startable processes of size $c+1$, then there exists a process $q \in Q$ such that $q$ writes to a register during $(D;\sigma_{q})$.
\end{lemma}
  \begin{proof}
	Let $X$ be the set of processes that own names in configuration $D$. Then processes in $X$ own names in range $\{1, \dots, |X| + c\}$. Let $Q = \{q_1, \dots, q_{c+1}\}$. 
	By way of contradiction assume that there is no process $q \in Q$ such that $q$ writes to a register during $(D; \sigma_q)$. Then for all $i$, $1 \leq i \leq c+1$, configurations $\sigma_{q_1}\dots\sigma_{q_{c+1}} (D)$ and $\sigma_{q_i}(D)$ are indistinguishable to $q_{i}$. 
	Let $q'$ be a startable process not in $Q$. 
	Hence, $\sigma_{q_1}\dots\sigma_{q_{c+1}} \sigma_{q'}(D)$ and $\sigma_{q'}(D)$ are indistinguishable to $q'$. 
	Therefore, all processes in $Q$ plus $q'$ return names from $\{1,\dots,|X|+ 1 + c\}$ 
    in execution $(D; \sigma_{q_1}\dots\sigma_{q_{c+1}}\sigma_{q'})$. 
		This is a contradiction because $|X| + c+2$ processes receive names from a set of size $|X| +  c+ 1$ implying that they cannot be assigned distinct names.
  \end{proof}

\begin{lemma}\label{lemma:add-main}\samepage
Let $D$ be a reachable configuration in which:
\begin{itemize}
	\item a set $Q$ of at least $c+1$ processes covers a set of $\ell \geq 1$ registers,
	\item there exists a set $Q' \subseteq Q$ of size $c+1$ such that no process in $Q'$ has written to a register and,
	\item there is a set of $c+1$ startable processes $P$, disjoint from $Q$. 
\end{itemize}
Then, there is a $P$-only schedule $\sigma_{P}$ such that at least $\ell+1$ registers are covered in $\sigma_{P}(D)$.
\end{lemma}

\begin{proof}
Let $X$ be the set of processes that own names in configuration $D$. Since no process in $Q'$ has written to a register, processes in $X$ own names in range $\{1, \dots, |X| + |Q| - |Q'| + c\}$. Let $L$ be the set of registers covered by $Q$ and $Q'' \subseteq Q$ be a set of processes that exactly covers $L$.
	Let $\widehat{\sigma_{P}}$ be a $P$-only schedule such that in execution $(D; \widehat{\sigma_{P}})$ all processes in $P$ complete their \getName{} operations. 	
	Then all processes in $P$ return names from $\{1,\dots,|X|+ |Q| - |Q'| + |P| + c\} = \{1,\dots,|X|+ |Q| + c\}$ in execution $(D; \widehat{\sigma_{P}})$. Suppose that in execution $(D; \widehat{\sigma_{P}})$, there is a process in $P$ that writes to a register not in $L$. Then let $\sigma_{P}$ be the shortest prefix of $\widehat{\sigma_{P}}$ such that a register $r \notin L$ is covered by a process in $P$. Hence in configuration $\sigma_{P}(D)$, $L$ is covered by $Q$ and $r$ is covered by $P$. Thus we are done. 
	Therefore, assume that in execution $(D; \widehat{\sigma_{P}})$ all processes in $P$ write only to $L$. Let $\pi_{Q''}$ be a block-write to $L$ by $Q''$. 		
Let $\sigma_{Q}$ be a $Q$-only schedule such that in execution $(\sigma_{P}\pi_{Q''}(D); \sigma_{Q})$ all processes in $Q$ complete their \getName{} operations.
Since configurations $\sigma_{P}\pi_{Q''}(D)$ and $\pi_{Q''}(D)$ are indistinguishable to all processes in $Q$, processes in $Q$ return names from $\{1,\dots,|X|+ |Q| + c\}$
     in execution $(\sigma_{P}\pi_{Q''}(D); \sigma_{Q})$.
    This is a contradiction because $|X| + |Q| + |P| = |X| + |Q| + c + 1$ processes receive names from a set of size $|X| + |Q| + c$ implying that they cannot be assigned distinct names.
\end{proof}

\begin{lemma}
\label{cor:one-shot}\samepage
Let $D$ be a reachable configuration in which:
\begin{itemize} 
	\item a set $P$ of processes exactly covers a set of $\ell \geq 1$ registers,
	\item there exists a process $q \in P$ such that $q$ has not written to any register and,
	\item there exists a set of $n' \geq c$ startable processes.
\end{itemize} 
Then, there exists a configuration reachable from $D$, in which at least $\ell + \lfloor \frac{n' - c}{c+1} \rfloor$ registers are covered.
\end{lemma}
\begin{proof}
Let $\Proc'$ be the set of all startable processes in $D$ and $Q \subseteq \Proc'$ be a set of $c$ processes. 
Then in configuration $D$, processes in $Q \cup P$ cover a set of $\ell$ registers where $|Q \cup P| \geq c+1$ and, 
	no process in set $Q \cup \{q\}$ has written to a register.  
	Hence using startable processes in $\Proc' \backslash Q$, we can inductively apply Lemma~\ref{lemma:add-main}, until we reach a configuration in which $\ell + \lfloor \frac{n' - c}{c+1} \rfloor$ registers are covered.
\end{proof}

\begin{lemma}\label{lemma:tight}
\label{corollary:one-shot-tight}\samepage
Let $A$ be an non-deterministic solo-terminating implementation of one-shot adaptive tight renaming. Let $D$ be any reachable quiet configuration with $n' \geq 2$ startable processes. Then there is an execution of $A$, starting from $D$ that requires at least $n'$ registers.
\end{lemma}
\begin{proof}
Let $p$ be a startable process and $\widehat{\sigma_p}$ be a $p$-only schedule such that in $(D; \widehat{\sigma_p})$, $p$ completes its \getName{}. Then by Lemma~\ref{lemma:add-help}, $p$ writes to a register. 
Let $\sigma_p$ be the shortest prefix of $\widehat{\sigma_p}$ such that in $(D; \sigma_p)$, $p$ covers a register. Then by Lemma~\ref{cor:one-shot}, there exists a configuration reachable from $\sigma_p(D)$, in which at least $1 + n' - 1 = n'$ registers are covered.
\end{proof}

In Lemma~\ref{lemma:final}, we show at least $\lfloor \frac{2(n' - c)}{c+2} \rfloor$ registers are required for a non-deterministic solo-terminating implementation of one-shot $(k+c)$-adaptive renaming starting from a quiet configuration with $n' \geq 2c+2$ startable processes. The intuition for this lemma is as follows. 
We prove the lemma by induction on $c$.  Starting from any quiet configuration, first we select a set $Q$ of $c+1$ startable processes such that one of them writes to a register in a solo-run and we stop it immediately before it writes. Then we choose a process $p$ not in $Q$ and run it until it covers a new register. If we succeed, we select another startable process not in $Q$.
We might not succeed for two reasons. First, we are out of startable processes in which case we are done.
Second, process $p$ completes its \getName{} and only writes to the set of covered registers. 
Then Lemma~\ref{cor:one-shot} provides a lower bound. 
Furthermore, starting from this configuration, if the set of covering processes perform a block-write and cover $p$'s trace, then no other process distinguishes this execution from one in which $p$ has not run at all. Therefore by Lemma~\ref{lem:domain-shrink}, our problem reduces to one-shot $(k+c-1)$-adaptive renaming. Hence we can invoke the induction hypothesis and conclude a second lower bound.
Our final lower bound is the maximum of these two lower bounds.

\begin{lemma}\label{lemma:final}
Let $A$ be a non-deterministic solo-terminating implementation of one-shot $(k+c)$-adaptive renaming. Let $D$ be any reachable quiet configuration with $n' \geq 2c+2$ startable processes. Then there is an execution of $A$,  starting from $D$ that requires at least $\lfloor \frac{2(n' - c)}{c+2} \rfloor$ registers.
\end{lemma}

\begin{proof}
Let $\Proc'$ be the set of startable processes in $D$. 
We prove the lemma by induction on $c$. For the base case $c = 0$, by Lemma~\ref{lemma:tight}, the hypothesis is true.
Suppose that the induction hypothesis is true for $c-1 \geq 0$. 
Since, $|\Proc'| > c+1$, by Lemma~\ref{lemma:add-help} there is a process $q \in \Proc'$ that writes to a register in a solo-execution starting from $D$. Let $\sigma_{q}$ be the shortest $q$-only schedule such that in configuration $\sigma_{q}(D)$, there is a register $r$ covered by $q$. Let $Q \subseteq (\Proc' \backslash \{q\})$ be a set of $c$ processes. 

First, we inductively construct a sequence of schedules $\delta_1,\delta_2,\ldots$ until we have constructed $\delta_{\ell}$ 
such that in $\delta_{\ell}(D)$ either 
\begin{compactenum}[a)]
	\item $|\Proc' \backslash Q|$ registers are covered or,
	\item a process $q'$ in $\Proc' \backslash Q$ has completed its \getName{} and has written only to registers covered by other processes.
\end{compactenum}

We maintain the invariant that for each $i \in \{1,\dots,\ell\}$ in configuration $\delta_i(D)$, 
a set $P_i \subseteq (\Proc' \backslash Q)$ of $i$ processes covers a set $L_i$ of $i$ distinct registers and $\delta_i$ is $P_i$-only.

Let $\delta_1$ be $\sigma_{q}$. Then in configuration $\delta_1(D)$, one register is covered, so the invariant is true for $P_1 = \{q\}$ and $L_1=\{r\}$.

Now consider $i \geq 1$.
If $a)$ or $b)$ holds for $\delta_i$, we let $\ell=i$ and the construction stops. Furtheremore in case $b)$, let $q'$ be the process that completes its \getName{}. 

Otherwise, since in $\delta_i(D)$ a set $L_i$ of $i$ distinct registers is covered, we construct $\delta_{i+1}$ as follows.
Select $p \in \Proc' \backslash (P_i \cup Q)$. Let $\gamma$ be the shortest $p$-only schedule such that either
\begin{compactenum}[1)]
	\item $p$ does a complete \getName{} in execution $(\delta_i(D); \gamma)$ and only writes to $L_i$, or
	\item in configuration $\gamma(\delta_i(D))$, $p$ covers a register $r' \notin L_i$.
\end{compactenum}

Let $\delta_{i+1}$ be $\delta_i \gamma$. 
First assume case $1)$ happens. Then condition $b)$ is achieved, the construction stops and we let $\delta_{\ell}=\delta_{i+1}$ and $q' = p$. 
Now suppose case $2)$ happens.
If $i+1 =|\Proc' \backslash Q|$, then condition $a)$ is achieved, the construction stops and $\delta_{\ell}=\delta_{i+1}$. 
Otherwise, the invariant remains satisfied for $L_{i+1}=L_{i} \cup \{r'\}$ and $P_{i+1}=P_i \cup \{p\}$.
Clearly, after at most $|\Proc' \backslash Q|$ steps either $a)$ or $b)$ is achieved.

In case $a)$, in configuration $\delta_{\ell}(D)$, $\ell = n' - c$ registers are covered so the lemma holds in this case.  
Now suppose case $b)$ happens. 
In configuration $\delta_{\ell - 1}(D)$, a set of $\ell - 1$ registers (i.e $L_{\ell - 1}$) are covered exactly by a set of processes $P_{\ell - 1}$ and process $q$ in $P_{\ell - 1}$ has not written to any registers. Furthermore in configuration $\delta_{\ell - 1}(D)$, all processes in $\Proc' \backslash P_{\ell - 1} \supseteq Q$ are startable. 
Therefore, by Lemma~\ref{cor:one-shot}, there is a configuration reachable from $\delta_{\ell - 1}(D)$ in which at least $\ell - 1 + \lfloor \frac{n' - \ell + 1 - c}{c+1} \rfloor$ registers are covered. 
Let $\pi_{P_{\ell - 1}}$ be a block-write by $P_{\ell - 1}$. Let $\sigma_{P_{\ell - 1}}$ be a $P_{\ell - 1}$-only schedule such that in execution $(\delta_{\ell} \pi_{P_{\ell - 1}} (D); \sigma_{P_{\ell - 1}}) $ all processes in $P_{\ell - 1}$ complete their \getName{}. Let $x$ be the name taken by $q'$ in execution $(D; \delta_{\ell})$. Let $X$ be the set of processes that own names in configuration $D$. Since in configuration $\delta_{\ell}(D)$, only processes in $X \cup P_{\ell}$ have invoked a \getName{}, $x \in \{1, \dots, |X| + \ell + c\}$. Note that $\delta_{\ell} \pi_{P_{\ell - 1}} \sigma_{P_{\ell - 1}} (D)$ and $\delta_{\ell - 1} \pi_{P_{\ell - 1}} \sigma_{P_{\ell - 1}} (D)$ are indistinguishable to all processes except $q'$. Hence in any $(\Proc' \backslash P_{\ell})$-only execution starting from $\delta_{\ell - 1} \pi_{P_{\ell - 1}} \sigma_{P_{\ell - 1}} (D)$, names returned by processes in $\Proc' \backslash P_{\ell}$ are in $\{1, \dots, k + c\} \backslash \{x\}$ where $k \geq |X| + \ell$ and therefore $\{x\} \subset \{1, \dots, k + c\}$. 
Thus by Lemma~\ref{lem:domain-shrink}, starting at $\delta_{\ell - 1} \pi_{P_{\ell - 1}} \sigma_{P_{\ell - 1}} (D)$, algorithm $A$ requires as many registers as a $(k+c-1)$-renaming algorithm starting at a quiet configuration with $|\Proc' - P_{\ell}| = n' - \ell$ startable processes. 
Therefore, by the induction hypothesis there is a reachable configuration from $\delta_{\ell - 1} \pi_{P_{\ell - 1}} \sigma_{P_{\ell - 1}} (D)$, in which at least $\lfloor \frac{2(n' - \ell - c + 1)}{c+1} \rfloor$ registers are covered.

The minimum of $\ell - 1 + \lfloor \frac{n' - \ell + 1 - c}{c+1} \rfloor$ and $\lfloor \frac{2(n' - \ell - c + 1)}{c+1} \rfloor$, is maximized when $\ell -1 + \lfloor \frac{n' - \ell + 1 - c}{c+1} \rfloor = \lfloor \frac{2(n' - \ell - c + 1)}{c+1} \rfloor$. Hence, $\ell = \lfloor \frac{n' - c}{c+2} \rfloor - 1$. Therefore the algorithm requires at least $\lfloor \frac{2(n' - c)}{c+2} \rfloor$ registers.
\end{proof}

\begin{theorem}\label{theorem:one-shot}\samepage
Any non-deterministic solo-terminating implementation of one-shot $(k+c)$-adaptive renaming requires at least $\lfloor \frac{2(n - c)}{c+2} \rfloor$ registers.
\end{theorem}
\begin{proof}
The initial configuration is a quiet configuration with $n$ startable processes. Hence, by Lemma~\ref{lemma:final}, there is an execution, starting from the initial configuration that requires at least $\lfloor \frac{2(n - c)}{c+2} \rfloor$ registers.
\end{proof}

Observe that by setting $c=0$, it follows from Theorem~\ref{theorem:one-shot} that any non-deterministic solo-terminating implementation of one-shot adaptive tight renaming requires $n$ registers. Since the number of startable processes is initially $n$, next corollary also follows from Lemma~\ref{lemma:final}. 

\begin{corollary}
Any non-deterministic solo-terminating implementation of one-shot adaptive tight renaming requires at least $n$ registers.
\end{corollary}

\section{Wait-Free One-shot $(b-1)$-Bounded $(k(k+1)/2)$-Adaptive Renaming}\label{Algorithms}
In this section we present a wait-free one-shot $(b-1)$-bounded $(k(k+1)/2)$-adaptive renaming algorithm from $b$ registers. Since $0$-bounded adaptive renaming is a trivial problem, we assume that $b \geq 2$.

The algorithms in this section employ a set  $\Reg=\{R[0], \ldots, R[b-1] \}$ of shared atomic registers. 
In our proofs, a \emph{register configuration} is a tuple $(V_0,\dots,V_{b-1})$, 
denoting that register $R[i]$, $0 \leq i\leq b-1$, has value $V_i$.
The proofs focus on just the sub-sequence of register configurations produced by an execution.
Specifically, given an execution $E=(C_0;\sigma)$, 
let \emph{write schedule} $\widehat{\sigma}$ be the sub-sequence of $\sigma$ that produces write steps in $(C_0;\sigma)$. 
Execution $E$ gives rise to the  sequence of \emph{register configurations} $\Gamma_E = C_0,C_{1},\dots$ 
such that the $i$-th step of $\widehat{\sigma}$ is a write that changes register configuration $C_{i-1}$ to register configuration $C_{i}$. 
For any scan operation $s$ in $E$, define $\Index{s}=i$, if $s$ occurs in $E$ between $C_i$ and $C_{i+1}$ in $\Gamma_E$. 
For any write operation $w$ in $E$, define $\Index{w}=i$, if $w$ is the $i$-th step of $\widehat{\sigma}$.
Notice that the view returned by a scan with index $i$ is equal to $C_i$.
A register configuration $C=(V_0,\dots,V_{b-1})$ is \emph{consistent} 
if $V_0=\dots=V_{b-1}$ in which case we say $V_0$ is the \emph{content} of $C$.
Let $C_i$ and $C_j$ be register configurations in the sequence $\Gamma_E=C_0,C_{1},\dots$ such that $i \leq j$. 
\emph{Interval}$[i, j]$ denotes the sub-sequence of steps in execution $E$ 
that begins at write operation $w$ where $\Index{w}=i$, 
and ends immediately after write operation $u$ where $\Index{u}=j$.
We use $\val{C}{R}$ to denote the content of register $R$ in configuration $C$. 
A local variable $x$ in these algorithms is denoted by $x_p$ 
when it is used in the method call invoked by process $p$.

\subsection{$(b-1)$-Bounded $(k(k+1)/2)$-Adaptive Renaming Using Atomic Scan}\label{WaitFree:Scan}

Fig.~\ref{fig:Ren-alg} presents a wait-free implementation of a one-shot $(b-1)$-bounded $(k(k+1)/2)$-adaptive renaming algorithm assuming an atomic scan operation.
In Section~\ref{sec:newScan}, we show how to remove this assumption. 

Each process maintains a set of processes, $S$, that it knows are participating including itself, 
and alternately executes write and scan operations.  
In the write operation, it writes $S$ 
to the next register after where it last wrote, in cyclic order through the $b$ registers.
After each of its scan operations, it updates 
$S$ to all the processes it sees in the scan together with the processes already in its set.
The process stops with an assigned name when either its scan shows exactly its own set, $S$, in every register, 
or $S$ has grown to size at least $b$.
If $|S|$ is less than $b$, its name is based on $|S|$ and its rank in $S$,
where $\mathrm{rank}(\mathit{id},S)=|\{i~|~(i \in S)\wedge(i \leq \mathit{id})\}|$. 
If $|S|$ is $b$ or greater, it returns a safe but large name.

\subsubsection*{Correctness and space complexity.} 
Since it is clear that the algorithm in Fig.~\ref{fig:Ren-alg} uses $b$ registers, 
the space complexity will follow immediately after confirming that it is a correct adaptive renaming algorithm.  
The correctness of this algorithm relies on the fact that 
if any two processes return names based on a set of size $s < b$, 
then they  have the same set. 
The main component of the proof is to establish this fact. 

\begin{figure}
  \textbf{shared:}
			$\Reg = R[0,\ldots, b - 1]$ is an array of multi-writer multi-reader registers, each register is initialized to $\emptyset$\\
	\textbf{local:}		
		An array $r[0,\dots,b-1]$; $\mathit{pos} \in \{0, \dots, b-1\}$ initialized to $0$; $S$ initialized to $\{\mathit{id}\}$;\\
		\setcounter{AlgoLine}{0}
		\begin{algorithm}[H]\caption{getName()}
		\Repeat{\label{line:wfr:repeat}$(|S| \geq b) \vee (r[0]=r[1]=\dots=r[b-1]=S)$ }{
      $R[\mathit{pos}].\text{write}(S)$\label{line:wfr:Write_S}\\
      $r[0,\dots,b-1] := \Reg.\text{scan}()$\label{line:wfr:scan}\\
      $S := \bigcup_{i=0}^{b-1}{r[i]} \cup S$\label{line:wfr:union}\\
      $\mathit{pos} := (\mathit{pos} + 1)$ \mod $b$\label{line:wfr:increment_pos}\\
    }\label{line:wfr:until}
    \uIf(\label{line:wfr:if}){$|S| \leq b-1$}{
      \Return{$(|S|(|S|-1))/2  + \mathrm{rank}(\mathit{id},S)$}\label{line:wfr:return1}\\ 
    }
    \Else{
      \Return{$b(b-1)/2 + \mathit{id}$}\label{line:wfr:return2}\\
    }
		\end{algorithm}

		\caption{$(b-1)$-Bounded $(k(k+1)/2)$-Adaptive Renaming Using Atomic Scan}\label{fig:Ren-alg}
\end{figure}

\begin{observation}\label{obs:simple}
For any write operation with value $S$ by process $p$, $p \in S$. 
\end{observation}

\begin{lemma}
\label{lem:regConfig}
For any execution $E$, let $C_a$ be a consistent register configuration with content $\widehat{S}$. 
For any register configuration $D$ following $C_a$ in $E$, 
define $\SmallRegs_{D}=\{R \in \Reg |\widehat{S} \not\subseteq \val{D}{R} \}$. 
Then there exists a one-to-one function $f_{D}:\SmallRegs_{D} \rightarrow \Proc$ satisfying, 
$\forall R \in \SmallRegs_{D}$, $f_{D}(R) \in \val{D}{R}$ 
and $f_{D}(R)$ performs at least one write in the execution interval between $C_a$ and $D$. 
\end{lemma}

\begin{proof}
Let $C_a,C_{a+1},\ldots$ be the sequence of register configurations that arises from $E$ starting at $C_a$.  
We prove the lemma by induction on the indices of this sequence.
The base case $\ell=a$, is trivially true since set $\SmallRegs_{C_a}$ is the empty set.  

Suppose that the induction hypothesis is true for $\ell-1 \geq a$. 
Let the write step between $C_{\ell-1}$ and $C_{\ell}$ be the operation, $w$, by process $p$, 
into register $\widehat{R}$ with value $V$.
Let $s$ be the most recent scan operation by $p$ preceding $w$ if it exists. 

If $\widehat{S} \subseteq V$, then  $\SmallRegs_{C_{\ell}} = \SmallRegs_{C_{\ell-1}} \setminus \{\widehat{R}\}$.
Define $f_{C_{\ell}}(R) = f_{C_{\ell-1}}(R)$, $\forall R \in \SmallRegs_{C_\ell}$.
Since $f_{C_{\ell-1}}$ satisfies the induction hypothesis, and
$\SmallRegs_{C_{\ell}} \subseteq  \SmallRegs_{C_{\ell-1}}$, $f_{C_{\ell}}$ also satisfies the induction hypothesis.

Now consider the case $\widehat{S} \not\subseteq V$. 
So $\SmallRegs_{C_{\ell}} = \SmallRegs_{C_{\ell-1}} \cup \{\widehat{R}\}$.
We first show that $s$ happens before $C_a$ or $w$ is the first write by $p$. 
Suppose, for the purpose of contradiction, that 
$a \leq \Index{s} \leq \ell-1$. 
We have $\forall R \in \Reg$,
$\widehat{S} \not\subseteq \val{C_{\Index{s}}}{R}$ 
since otherwise, by Line~\ref{line:wfr:union}, $\widehat{S} \subseteq V$. 
Thus $|\SmallRegs_{C_{\Index{s}}}|=b$. 
By the induction hypothesis, $f_{C_{\Index{s}}}$ selects a distinct process from each register in $\SmallRegs_{C_{\Index{s}}}$, 
implying, 
by Line~\ref{line:wfr:union}, that the size of $S_p$ is at least $b$. 
Hence $p$ would have stopped in Line~\ref{line:wfr:until} before performing any write operation.  
Therefore $s$ happens before $C_a$ or $w$ is the first write by $p$, 
and consequently any write by $p$ before $w$ happens before $C_a$. 
On the other hand,  $\forall R \in \SmallRegs_{C_{\ell-1}}$, $f_{C_{\ell-1}}(R)$ performs a write during Interval$[a,\ell-1]$ 
implying $p$ is not in $\{f_{C_{\ell-1}}(R)~|~ R \in \SmallRegs_{C_{\ell-1}}\}$.
By Observation~\ref{obs:simple}, $p \in V$ and $p$ performs a write after $C_a$. 
Therefore by defining $f_{C_{\ell}}(R)=f_{C_{\ell-1}}(R)$, $\forall R \in (\SmallRegs_{C_{\ell-1}} \setminus \{\widehat{R}\})$ 
and $f_{C_{\ell}}(\widehat{R}) = p$, the induction hypothesis holds for $\ell$. 
\end{proof}

\begin{lemma}
\label{lem:LA}
For any execution $E$, let $\widehat{S_p}$ and $\widehat{S_q}$ be the value of $S_p$ and $S_q$ in Line~\ref{line:wfr:if} for $p$ and $q$ 
when they have completed the repeat loop. 
If $|\widehat{S_p}|=|\widehat{S_q}| < b$ then $\widehat{S_p} = \widehat{S_q}$.
\end{lemma}

\begin{proof}
Let $C_p$  and $C_q$ be the consistent register configurations that resulted in $\widehat{S_p}$ and $\widehat{S_q}$ respectively and assume, 
without loss of generality, that $C_p$ precedes $C_q$ in $\Gamma_E$.
By Line~\ref{line:wfr:until}, 
$R[0]=\dots=R[b-1]=\widehat{S_q}$ in $C_q$. 
Thus, either $\forall R \in \Reg$,  $\widehat{S_p} \subseteq \val{C_q}{R}$ or 
$\forall R \in \Reg$, $\widehat{S_p} \not\subseteq\val{C_q}{R}$. 

For the first case, by Line~\ref{line:wfr:union}, 
$\widehat{S_p} \subseteq \widehat{S_q}$ and since $|\widehat{S_p}|=|\widehat{S_q}|$, 
$\widehat{S_p} = \widehat{S_q}$. 
For the latter case, set $\SmallRegs_{C_{q}}=\{R \in \Reg ~|~\widehat{S_p} \not\subseteq \val{C_q}{R}\}$ has size $b$. 
By Lemma~\ref{lem:regConfig}, 
there is a distinct process in each register in $\SmallRegs_{C_q}$. 
So there are at least $b$ distinct processes in $\widehat{S_q}$ contradicting $|\widehat{S_q}| < b$.
\end{proof}

\begin{lemma}\label{lem:wfub:name-distinctness}
The names returned by any two distinct processes are distinct.
\end{lemma}

\begin{proof}
Let $\widehat{S_p}$ and $\widehat{S_q}$ be the values of $S_p$ and $S_q$ in Line~\ref{line:wfr:if}. 
Without loss of generality, assume that $|\widehat{S_p}| \leq |\widehat{S_q}|$. 
If $|\widehat{S_p}| \geq b$ and $|\widehat{S_q}| \geq b$, 
the names returned by $p$ and $q$ in Line~\ref{line:wfr:return2} are distinct because $p \neq q$.
If $|\widehat{S_p}| < b$ and $|\widehat{S_q}| \geq b$, 
then, by Line~\ref{line:wfr:return1}, the name returned by $p$ is at most $(b-1)(b-2)/2 + (b-1) = b(b-1)/2$ 
and, by Line~\ref{line:wfr:return2}, the name returned by $q$ is bigger than $b(b-1)/2$. 
If $|\widehat{S_p}|  < b$ and $|\widehat{S_q}| < b$, 
both processes return at Line~\ref{line:wfr:return1}.
First suppose $\ell = |\widehat{S_p}|<|\widehat{S_q}| $.  
Then the name returned by $p$ is at most $(\ell+1)(\ell)/2$ 
and the name returned by $q$ is at least $(\ell+1)(\ell)/2 + 1$. 
If $|\widehat{S_p}|=|\widehat{S_q}|$, by Lemma~\ref{lem:LA}, $\widehat{S_p} = \widehat{S_q}$. 
Therefore $\mathrm{rank}(p, \widehat{S_p}) \neq \mathrm{rank}(q, \widehat{S_q})$. 
Thus, in all cases the names returned by $p$ and $q$ are distinct.
\end{proof}

\begin{observation}\label{obs:obvious}
Set $\{p\}$ is written by $p$ before any other write of any set $V\supseteq\{p\}$.
\end{observation}

\begin{lemma}\label{lem:wfub:namespace-range}
Let $k$ be the number of participating processes during process $p$'s \getName{}.  
Then, any name returned by $p$ is in the range $\{1,\ldots,\frac{k(k+1)}{2}\}$ if $k < b$ 
and in the range $\{1 ,\ldots,n+\frac{b(b-1)}{2}\}$ if $k \geq b$. 
\end{lemma}

\begin{proof}
By Observation~\ref{obs:obvious}, $\forall q \in S_p$, $q$ performs at least one write before $p$ returns. Thus, $\forall q \in S_p$, $q$ is a participating process. Hence, $|S_p| \leq k$. 
If $k < b$, then $|S_p| < b$. Therefore, process $p$ returns in Line~\ref{line:wfr:return1}, and the name is in the range $\{1,\ldots,\frac{k(k+1)}{2}\}$.
If $k \geq b$, then $p$ returns either in Line~\ref{line:wfr:return1} or in Line~\ref{line:wfr:return2}. Therefore the name is in the range $\{1,\ldots,\frac{b(b-1)}{2}+n\}$.
\end{proof}

In summary, Lemmas \ref{lem:wfub:name-distinctness} and \ref{lem:wfub:namespace-range} imply that 
the  algorithm in Fig.~\ref{fig:Ren-alg} is an $(b-1)$-bounded $(k(k+1)/2)$-adaptive renaming algorithm that uses $b$ registers assuming the availability of the atomic scan operation.

\subsubsection*{Step complexity.}
We now bound the maximum number of steps (scans and writes) that any process can take during its
execution of \getName. 
Lemma \ref{lem:kcomplete} establishes the most important piece of the step complexity of the algorithm in Fig.~\ref{fig:Ren-alg}. 
In this lemma we prove that if there exists a register configuration 
in which there are at least $k$ registers, each of which contains a set of size at least $k$, 
then the number of distinct process names in any subsequent scan is at least $k$. 
We call such a register configuration \emph{\complete{k}}, 
and any set of such registers is a \emph{\fullset{k}}.
The core idea is that after a \complete{k} configuration with \fullset{k} $ \Reg'$, 
every write with set size less than $k$ to a register in $\Reg'$ is performed by a distinct writer. 
It then follows that the union of the sets appearing in $\Reg'$ always will have  size at least $k$. 
For the proof, given set of registers $\Reg'\subseteq \Reg$ and a register configuration $D$, 
we will be interested in those registers in $\Reg'$ that contain a set smaller than $k$,
and in the processes that wrote these small sets to these registers.  
Therefore,  define $\rho_{\Reg'}(D, k) = \{ R \in \Reg' ~\big\vert~ |\val{D}{R}|   < k \}$. 
Let \emph{\writer{j}{R}} denote the process that performs the most recent write to register $R$ 
preceding register configuration $C_{j}$.
For any set of registers $\Reg'$, register configuration $C_j$ and an integer $k$, 
define $W_{\Reg'}(j, k) = \{ \text{\writer{j}{R} }|\text{ } R \in \rho_{\Reg'}(C_{j}, k)\}$. 
Notice that a register configuration $D$ is \complete{k} if there exists a set $\Reg'$ of $k$ registers 
where $\rho_{\Reg'}(D, k) = \emptyset$. 
Furthermore, $\Reg'$ is \fullset{k} at register configuration $D$. 

Lemma \ref{lem:kcomplete} uses a proof structure that is more elaborate than, but reminiscent of, that of Lemma \ref{lem:regConfig}. 

\begin{lemma}\label{lem:kcomplete}
For any execution $E$, let $C_a$ be a \complete{k} register configuration 
where $0 \leq k \leq b-1$ and let $\Reg'$ be a \fullset{k} of $C_a$. 
For any register configuration $C_e$ following $C_a$ in $\Gamma_E$, 
there exists a one-to-one and onto function $g_{C_e}: \rho_{\Reg'}(C_e, k) \rightarrow W_{\Reg'}(e, k)$
satisfying, 
$\forall R \in \rho_{\Reg'}(C_e, k)$, $g_{C_e}(R) \in \val{C_e}{R}$. 
Furthermore, each process in $W_{\Reg'}(e, k)$
performs at least one write in Interval$[a, e]$.  
\end{lemma}
   
\begin{proof}
Let $C_a,C_{a+1},\ldots$ be the sequence of all register configurations starting at $C_a$.  
We prove the lemma by induction on the indices of this sequence.
The base case $\ell=a$, is trivially true since set $\rho_{\Reg'}(C_a, k)$ is the empty set. 

Suppose that the induction hypothesis is true for $\ell-1 \geq a$. 
Let the write step between $C_{\ell-1}$ and $C_{\ell}$ be the operation, $w$, by process $q$, 
into register $\widehat{R}$ with value $V$.
Let $s$ be the most recent scan operation by $q$ preceding $w$ if it exists. 

Suppose that 
$\widehat{R} \notin {\Reg'}  $.
Then $ \rho_{\Reg'}(C_{\ell-1} , k) =  \rho_{\Reg'}(C_{\ell} , k)$ 
and $W_{\Reg'}(\ell-1, k) = W_{\Reg'}(\ell, k)$, so the induction hypothesis 
holds trivially for $\ell$ by setting $g_{C_{\ell}} = g_{C_{\ell -1}}$.

Suppose that  $\widehat{R} \in \Reg'$ and  $|V| \geq k$.  
Then  $\rho_{\Reg'}(C_{\ell}, k) =  \rho_{\Reg'}(C_{\ell-1}, k) \setminus \{\widehat{R}\}$ 
and $W_{\Reg'}(\ell, k) = W_{\Reg'}(\ell-1 , k) \setminus \{ \writer{\ell-1}{\widehat{R}} \}$.  
So  the hypothesis holds for $\ell$ by setting $g_{C_{\ell}} = g_{C_{\ell -1}}$ for each $R \in \rho_{\Reg'}(C_{\ell}, k)$.

Finally, consider the case $\widehat{R} \in \Reg'$ and $|V| < k$. 
We first show that $s$ happens before $C_a$ or $w$ is the first write by $q$.  
Suppose, for the purpose of contradiction, that 
$s$ happens after $C_a$.  
Then $a \leq \Index{s} \leq \ell-1$. 
For each $R \in \Reg'$, we have  $|\val{C_{\Index{s}}}{R}| < k$ 
since otherwise, by Line~\ref{line:wfr:union}, $|V| \geq k$. 
Thus $|\rho_{\Reg'}(C_{\Index{s}}, k)| = k$. 
By the induction hypothesis, $g_{C_{\Index{s}}}$ is a bijection, so
 $|\rho_{\Reg'}(C_{\Index{s}}, k)| = |W_{\Reg'}(\Index{s}, k)| = k$,
and 
 $\forall R \in \rho_{\Reg'}(C_{\Index{s}}, k)$, $g_{C_{\Index{s}}}(R) \in \val{C_{\Index{s}}}{R}$.   
Therefore,
by Line~\ref{line:wfr:union},  the size of $S_q$, and hence the size of $V$, is at least $k$, which is a contradiction. 

Therefore $s$ happens before $C_a$ or $w$ is the first write by $q$, 
and consequently any write by $q$ before $w$ happens before $C_a$. 
On the other hand, by the induction hypothesis,
 $\forall R \in \rho_{\Reg'}(C_{\ell-1}, k)$, \writer{\ell-1}{R} performs a write during Interval$[a,\ell-1]$ 
implying $q$ is not in $W_{\Reg'}(\ell-1 , k)$.
We have
$\rho_{\Reg'}(C_{\ell}, k) = \rho_{\Reg'}(C_{\ell-1}, k) \cup \{\widehat{R}\}$
and
$W_{\Reg'}(\ell, k) = W_{\Reg'}(\ell-1 , k) \setminus{\writer{\ell-1}{\widehat{R} }} \cup \{q\}$,  
whether or not 
$\widehat{R}$ is in $\rho_{\Reg'}(C_{\ell-1}, k)$. 
Furthermore, $q$ performs a write after $C_a$.
Therefore, the induction hypothesis holds for $\ell$ by defining
$g_{C_{\ell}} = g_{C_{\ell -1}}$ for each $R \in \rho_{\Reg'}(C_{\ell}, k)  \setminus \{\widehat{R}\}$,
and $g_{C_{\ell}}(\widehat{R}) = q$.

\end{proof}

\begin{lemma}\label{lem:secondWrite}
Let $D$ be a \complete{k} register configuration. 
Then for each process $p$ in $\Proc$, 
$p$'s second write after $D$ if it exists, has a value with size at least $k$.
\end{lemma}
\begin{proof}
Let $w$ be the second write operation by $p$ after $D$ if it exists. 
Let $s$ be the most recent scan operation by $p$ preceding $w$. 
Since $w$ is the second write by $p$, the value returned by $s$ is equal to a register configuration $D'$ following $D$. 
Let $\Reg'$ be the \fullset{k} of $D$. 
If  $\exists R \in  \Reg$,  $|\val{{D'}}{R} | \geq k$, then
by Line~\ref{line:wfr:union}, the size of $S_p$ at $s$ is at least $k$.
Otherwise, all registers, and hence all registers in $\Reg'$, contain sets of size less than $k$.
Therefore, $|\rho_{\Reg'}(D', k)| = k$. 
So,   by Lemma~\ref{lem:kcomplete}, $|W_{\Reg'}(\Index{D'}, k)| = k$. 
Thus, again by Line~\ref{line:wfr:union}, the size of $S_p$ at $s$ is at least $k$. 
Hence $w$ has a value with size at least $k$.

\end{proof}

\begin{lemma}\label{lem:mWrite}
Let $E$ be an execution whose first operation is a write by $p$ and contains the next $b$ scans by $p$. 
Furthermore, during $E$, every write by $p$ has value $Q$ and no write has value $Q' \subsetneq Q$. 
Then $p$ either terminates or the size of $S_p$ increases.
\end{lemma}

\begin{proof}
By Lines~\ref{line:wfr:Write_S}-\ref{line:wfr:scan}, each scan operation is preceded by a write operation. Hence $E$ contains $b$ writes by $p$. 
Therefore, during $E$, $p$ writes $Q$ to all $b$ registers. 
Let $(V_0, \dots, V_{b-1})$ be the value returned by $p$'s last scan during $E$. 
Because $E$ does not contain any write with value $Q' \subsetneq Q$ either $V_0 = \dots= V_{b-1} = Q$ in which case $p$ terminates or $\exists i$, $0 \leq i \leq b-1$ such that $V_i \not\subseteq Q$. 
In the latter case by Line~\ref{line:wfr:union}, the size of $S_p$ increases.

\end{proof}

\begin{lemma}\label{lem:mkWrites}
Let $D$ be a \complete{k} register configuration where $0 \leq k \leq b-1$. 
Then for each process $p$ in $\Proc$, $p$'s $(bk+2)$-nd write after $D$, if it exists, has a value with size at least $k+1$.
\end{lemma}

\begin{proof}
Let $w$ be the second write by $p$ after $D$. 
Suppose that $p$ writes $Q$ at $w$. 
By Lemma~\ref{lem:secondWrite}, $|Q| \geq k$. 
If $|Q |\geq  k +1 $ or $p$ terminates before writing $bk$ more times, we are done. 
Therefore, suppose that $|Q | = k$ and $p$ performs $bk$ writes after $w$.
Then $|Q \setminus\{p\}| = k-1$. 
By Lemma~\ref{lem:secondWrite}, after $D$, $\forall q \in Q$, 
$q$ writes a value with size smaller than $k$ at most once. 

Let $E$ be the execution whose first operation is $w$ and contains the next $bk$ scan operations by $p$. 
Partition $E$ into disjoint segments, $E = (E_1, \dots, E_k)$,
satisfying $\forall \ell$, $1 \leq \ell \leq k$, the first operation in $E_{\ell}$ is a write operation by $p$ and $E_{\ell}$ contains the next $b$ scans by $p$. 
Notice that $E$ contains exactly $bk$ write operations by $p$ and since $w$ is the first operation of $E$, $p$ performs at least one more write after $E$ ends.    
Since there are at most $|Q \setminus \{p\}| = k-1$ writes after $w$ that have a value $V$ satisfying  
$V \subsetneq Q$, there exists an $\ell$, $1 \leq \ell \leq k$ such that all writes during $E_{\ell}$ 
have a value that is not a proper subset of $Q$. 
Since $p$ does not terminate during $E_{\ell}$, by Lemma~\ref{lem:mWrite}, the size of $S_p$ after $E_{\ell}$ (hence, after $E$) is at least $k+1$. Hence, $p$'s $(bk+2)$-nd write after $D$ has a value with size at least $k+1$.

\end{proof}

\begin{lemma}\label{lem:interrupt}
For any execution $E$ in which $p$ does not terminate, 
let $O$ be the set of all scan operations by $p$ during $E$. 
Let $Z = \{ \writer{\Index{s}}{R}~|~s \in O \text{ and } R \in \Reg\}$.  
Then $|Z| < b $.
\end{lemma}

\begin{proof}
For any scan $s \in O$ and any register $R$, $\writer{\Index{s}}{R} \in \val{C_{\Index{s}}}{R}$. 
Therefore, by Line~\ref{line:wfr:union}, for any $s \in O$, after $s$, $S_p$ contains $\writer{\Index{s}}{R}$.  
Since $p$ does not terminate after $s$, at $s$, $|S_p| < b$. 
Hence $|Z| < b$.

\end{proof}

\begin{lemma}\label{lem:kComplete}
Let $D$ be a \complete{k} register configuration where $0 \leq k \leq b-1$. 
Then for each process $p$ in $\Proc$, $p$ makes at most $bk + 1  + b (\frac{(b-1)(bk+1)}{b-k} + 1)$ write operations before it terminates or a \complete{(k+1)} register configuration is achieved.
\end{lemma}

\begin{proof}
By Lemma~\ref{lem:interrupt}, $p$'s $(bk+2)$-nd write, say $w$, after $D$ has a value with size at least $k+1$.
Let $E$ be the execution whose first operation is $w$ and contains the next $b (\frac{(b-1)(bk+1)}{b-k} + 1)$ scan operations by $p$. 
Partition $E$ into segments, $E = (E_1, \dots, E_{\frac{(b-1)(bk+1)}{b-k} + 1})$,
satisfying $\forall \ell$, $1 \leq \ell \leq \frac{(b-1)(bk+1)}{b-k} + 1$, the first operation in $E_{\ell}$ is a write operation by $p$ and $E_{\ell}$ contains the next $b$ scans by $p$. 
Let $O$ be the set of all scan operations by $p$ during $E$. 
Let $Z = \{ \writer{\Index{s}}{R}~|~s \in O \text{ and } R \in \Reg\}$.
Let $U$ be the set of all write operations by processes in $Z$ during $E$ such that $\forall u \in U$, the value of $u$ has a size smaller than or equal to $k$. 
By Lemmas~\ref{lem:mkWrites} and \ref{lem:interrupt}, $|U| \leq (bk+1)|Z| \leq (b-1)(bk + 1)$. 
Let $U_{\ell} = \{u ~|~ u \in U \text{ and } u \text{ happens during }E_{\ell}\}$. 

By the pigeon whole principle, there exists an $\ell$ such that $|U_{\ell}| < b-k$.
Let $s_{\ell}$ be $p$'s last scan during $E_{\ell}$.  
Since during $E_{\ell}$, $p$ writes a value with size at least $k+1$ to all $b$ registers 
and the number of writes with value smaller than $k+1$ and scanned by $p$ (i.e $|U_{\ell}|$), is less than $b-k$, $s_{\ell}$ returns a view in which at least $k+1$ registers have size at least $k+1$. 
Hence, $C_{\Index{s_{\ell}}}$ is \complete{(k+1)}. 
 
\end{proof}

\begin{lemma}\label{lem:allWrites}
No process writes more than $3b^4 \ln{b} $ times.
\end{lemma}

\begin{proof}
	
By Lemma \ref{lem:kComplete}, a process can write at most
$bk + 1  + b (\frac{(b-1)(bk+1)}{b-k} + 1)$ times between a \complete{k} and a \complete{(k+1)} configuration.	
The initial configuration is \complete{0} and an \complete{b} configuration cannot exist.
Therefore a process can write a most 
\begin{eqnarray*}
 & &\sum_{k=0}^{b-1} {\big(bk + 1  + b (\frac{(b-1)(bk+1)}{b-k} + 1)\big)} \\ 
 & = & \sum_{k=0}^{b-1} {(bk + 1  + b )} + b(b-1) \sum_{k=0}^{b-1} {(\frac{bk+1}{b-k})}\\
  & = &\sum_{k=0}^{b-1} {(bk + 1  + b )} + b(b-1) \sum_{k=1}^{b} {\big(\frac{b^2 - bk + 1}{k}\big)} \\
 & < & \frac{b^3}{2} +  b -1 + b^2  + b^2 \sum_{k=1}^{b} {(\frac{b^2}{k} - b + \frac{1}{k})} < 3b^4 \ln{b} 
\end{eqnarray*}
times before it terminates.

\end{proof}

\subsection{$(b-1)$-Bounded $(k(k+1)/2)$-Adaptive Renaming Using Registers}\label{sec:newScan}
 We replace the atomic scan in Fig.~\ref{fig:Ren-alg} with a new function, \newScan{}, and the \getName{} algorithm also changes accordingly. 
The revised renaming algorithm is shown in Fig.~\ref{fig:Ren-Sacn-alg}.  
In the \getName{} algorithm, processes augment the values they write to each register with their ids and sequence numbers in order to guarantee the uniqueness of the value of each write. 
This prevents the ABA problem. 
Each register $R \in \Reg$ stores an ordered triple $(\mathit{set}, \mathit{id}, \mathit{seqNumber})$.

During a \newScan{} operation by process $p$, $p$ performs a \collect{$\Reg$} in Line~\ref{line:scan:collect}, by reading $R[0]$ through $R[b-1]$ consecutively and returns a \emph{collect}. After each \collect{$\Reg$}, $p$ updates its set $S$ from this collect. It repeatedly gets a collect until either the size of set $S$ becomes at least $b$ or it obtains two identical consecutive collects and returns this collect. 
If \newScan{} terminates at Line~\ref{line:scan:success}, 
then the returned collect is equivalent to the returned value of a linearizable implementation of a scan \cite{GLS1992a,IP1992a,IL1993a,DS1997a,DW1999a}. 
Hence all the proofs in Section~\ref{WaitFree:Scan} hold when \newScan{} terminates at Line~\ref{line:scan:success}. 
So to establish the correctness, we need to prove that when a process $p$ returns in Line~\ref{line:scan:abort}, in fact more than $b - 1$ processes are participating, hence the name returned by $p$'s \getName{} is valid. 
This is shown in Lemma~\ref{lem:abort}.  
In Lemma~\ref{lem:allWrites}, we showed that the number of writes by each process is bounded. 
Since the sequence number $\mathit{seqNumber}$, cannot get larger than the number of write operations by each process, the size of each register is also bounded. 
Therefore, after Lemma \ref{lem:abort} it will remain to prove that the \getName{} algorithm in Fig.~\ref{fig:Ren-Sacn-alg} is wait-free. 
This will be established in Lemmas~\ref{lem:3mRead} through \ref{lem:allScans}, 
by bounding the number of steps of each \newScan{} operation. 

For any read operation $o$ of register $R$, define $\writeOp{o}$ to be the most recent write operation to $R$ preceding $o$ if it exists and $\bot$ otherwise. For any write operation $w$, let $\mathrm{performer}(w)$ denote the process that performs $w$. 
For any set of write operations $W$, let $\writers{W} = \{\mathrm{performer}(w)~|~ w \in W\}$.

\begin{lemma}\label{lem:abort}
Let $k$ be the number of participating processes during process $p$'s \getName{}.
If a \newScan{} operation by $p$ returns in Line~\ref{line:scan:abort}, then $k \geq b$.
\end{lemma}

\begin{proof}
Let $\widehat{S_p}$ be the value of $S_p$ when $p$'s \getName{} returns. 
Since $p$'s \newScan{} operation returns in Line~\ref{line:scan:abort}, $|\widehat{S_p}| \geq b$. 
By Observation~\ref{obs:obvious}, $\forall q \in \widehat{S_p}$, $q$ 
performs at least one write before $p$ returns. Thus, $\forall q \in \widehat{S_p}$, $q$ has invoked a \getName{} before $p$ returns. Therefore, $k \geq |\widehat{S_p}| \geq b$. 
\end{proof}

\begin{figure}
  \textbf{shared:}
			$R[0\ldots b-1]$ is an array of multi-writer multi-reader registers, each register is initialized to $(\emptyset, 0, 0)$;\\
	\textbf{local:}		
		 $r[0,\dots,b-1]$; $\mathit{pos} \in \{0, \dots, b-1\}$ initialized to 0; $\mathit{seqNumber}$ is a non-negative integer initialized to $0$; $S$ is initialized to $\{\mathit{id}\}$; $\mathit{largeSet}$ is a boolean; \\
		\setcounter{AlgoLine}{0}
		\begin{algorithm}[H]\caption{getName()}
		\Repeat{$(|S| \geq b) \vee (r[0].\mathit{set}=r[1].\mathit{set}=\dots=r[b-1].\mathit{set}=S)$}{
			$\mathit{seqNumber} = \mathit{seqNumber} + 1$\\
      $R[\mathit{pos}].\text{write}(S, \mathit{id}, \mathit{seqNumber})$\\
      ($\mathit{largeSet}$, $r[0,\dots,b-1]$) $:= \Reg.\newScan{S}$\\
      \If{$\mathit{largeSet}$}{\Return{$b(b-1)/2 + \mathit{id}$}\label{line:largeReturn}}
			$S := \bigcup_{i=0}^{b-1}{r[i].\mathit{set}}$\\
      $\mathit{pos} := (\mathit{pos} + 1)$ \mod $b$\\
    }
    \uIf{$|S| \leq b-1$}{
      \Return{$(|S|(|S|-1))/2  + \mathrm{rank}(\mathit{id},S)$}\\ 
    }
    \Else{
      \Return{$b(b-1)/2 + \mathit{id}$}\\
    }
		\end{algorithm}
		\textbf{local:}		
		$a[0\ldots b-1]$; $a'[0\ldots b-1]$ each element is initialized to $(\emptyset, 0, 0)$; 
		\begin{algorithm}[H]\caption{newScan($S$)}
		\Repeat{$a = a'$\label{line:scan:repeat}}{
      $a := a'$\\
			$a' := $\collect{$\Reg$}\label{line:scan:collect}\\
			$S := \bigcup_{i=0}^{b-1}{a'[i].\mathit{set}} \cup S$\label{line:scan:union}\\
			\If(\label{line:scan:terminate_b}){$|S| > b-1$ \label{line:scan:if-abort}}{
			\Return{(\True, $a'$)}\label{line:scan:abort}}
			}
			\Return{(\False, $a'$)}\label{line:scan:success}\\		
		\end{algorithm}
		\caption{$(b-1)$-Bounded $(k(k+1)/2)$-Adaptive Renaming Using Registers}\label{fig:Ren-Sacn-alg}
\end{figure}

\begin{lemma}\label{lem:3mRead}
Let $E$ be an execution such that any step by process $p$ during $E$ is part of a single \newScan{} operation. 
If $E$ contains at least $3b$ reads by $p$ and does not contain any write operations, 
then $p$'s \newScan{} terminates during $E$. 
\end{lemma}

\begin{proof}
Since $E$ contains no write operation, every $3b$ reads by $p$ must contain two complete identical collects. 
Hence, $p$ must terminate due to Line~\ref{line:scan:repeat}.

\end{proof}

\begin{lemma}\label{lem:mWriter}
Let $E$ be an execution such that any step by process $p$ during $E$ is part of a single \newScan{} operation, $s$.  
Let $O$ be the set of all reads by $p$ during $E$ and $W = \{ \writeOp{o} ~|~ o \in O \} \backslash \{ \bot \}$. 
If $|\writers{W}| \geq b$, then $s$ contains at most $2b$ read operations after $E$ ends. 
\end{lemma}

\begin{proof}
If $p$ performs fewer than $2b$ read operations after $E$ ends, we are done.
Let $E'$ be an execution which starts after $E$ ends and contains $2b$ reads by $p$. 
Since every $2b$ reads by $p$ must contain a complete collect, after $p$'s complete collect during $E'$, by Line~\ref{line:scan:union}, $S_p$ includes all processes in $\writers{W}$. 
Hence, after $p$'s complete collect during $E'$, $|S_p| \geq |\writers{W}| \geq b$. 
Therefore, by Line~\ref{line:scan:terminate_b}, $s$ must terminate.
\end{proof}

\begin{lemma}\label{lem:allScans}
No \newScan{} operation contains more than $10b^6\ln{b}$ reads.
\end{lemma}
\begin{proof}
By way of contradiction, let $E$ be an execution in which process $p$ performs a single \newScan{} $s$, and it contains more than $10b^6 \ln{b}$ reads. 
Let $E'$ be a prefix of $E$ that contains $9b^6 \ln{b}$ reads by $p$. 
Partition $E'$ into disjoint segments, $E' = (E_1, \dots, E_{3b^5 \ln{b}})$, 
satisfying $\forall \ell$, $1 \leq \ell \leq 3b^5 \ln{b}$, $E_{\ell}$ contains $3b$ reads by $p$.  
Let $O$ be the set of all read operations by $p$ during $E'$ and $W = \{ \writeOp{o} ~|~ o \in O \} \backslash \{ \bot \}$.

Suppose there is an $\ell$ such that $E_{\ell}$ contains no write operation in $W$. This implies that $E_{\ell}$ contains no write operation. 
Therefore by Lemma~\ref{lem:3mRead}, $p$ terminates $s$ during $E_{\ell}$. 

Otherwise, each $E_{\ell}$ contains at least one write in $W$. Hence $|W| \geq 3b^{5} \ln{b}$.   
Since by Lemma~\ref{lem:allWrites}, each process $p$ performs at most $3b^{4} \ln{b}$ writes, $|\writers{W}| \geq b$. 
Therefore, by Lemma~\ref{lem:mWriter}, $s$ contains at most $2b$ reads after $E'$ ends. Hence $E$ contains at most $9b^6\ln{b} + 2b < 10b^6\ln{b}$ reads by $p$.
\end{proof}

\begin{lemma}
No process performs more than $31b^{10} \ln^2{b}$ shared steps (read or write).
\end{lemma}

\begin{proof}
By Lemma~\ref{lem:allWrites}, each process $p$ performs at most $3b^{4} \ln{b}$ writes. Hence, $p$ performs 
 at most $3b^{4} \ln{b}$ \newScan{} operations. By Lemma~\ref{lem:allScans}, $p$ performs at most $10b^6\ln{b}$ reads in each \newScan{} operation. Hence $p$ performs at most $3b^{4}\ln{b} + (3b^{4} \ln{b})(10b^6\ln{b}) \leq  31 b^{10} \ln^2{b}$ shared steps.
\end{proof}

\begin{theorem}
  For any $b\geq 2$, there is a wait-free one-shot $(b-1)$-bounded $(k(k+1)/2)$-adaptive renaming algorithm 
  implemented from $b$ bounded registers.
  Additionally, when $k \geq b$, the returned names are in the range $\{1,\ldots,n+\frac{b(b-1)}{2}\}$.
\end{theorem}
Setting $b = \lceil \sqrt{n} \rceil + 1$, we have a wait-free one-shot $\lceil \sqrt{n} \rceil$-bounded $(k(k+1)/2)$-adaptive renaming algorithm from $\lceil \sqrt{n} \rceil + 1$ bounded registers. This implies that the algorithm returns names in the range $\{1, \ldots, (k(k+1)/2)\}$ when $k \leq \lceil \sqrt{n} \rceil$, and returns names in the range $\{1,\ldots, n + \frac{\lceil \sqrt{n} \rceil(\lceil \sqrt{n} \rceil + 1)}{2}  \}$ when $k \geq \lceil \sqrt{n} \rceil + 1$. Note that when $k \geq \lceil \sqrt{n} \rceil + 1$, $ n + \frac{\lceil \sqrt{n} \rceil(\lceil \sqrt{n} \rceil + 1)}{2} \leq k^2 + \frac{k^2}{2}$. Hence, $\forall k \in \{1,\dots, n \}$, the algorithm returns names in the range $\{ 1,\ldots,(3k^2)/2 \}$. 
\begin{corollary}
There is a wait-free one-shot $(3k^2)/2)$-adaptive renaming algorithm implemented from $\lceil \sqrt{n} \rceil + 1$ bounded registers.
\end{corollary}

\section{Obstruction-Free $(b-1)$-Bounded $k$-Adaptive Renaming}\label{sec:ObsUB}

\begin{figure}
  \textbf{shared:}
			$\Reg = R[1,\ldots, b]$ is an array of multi-writer multi-reader registers, each register is initialized to $(\emptyset, \bot, 1)$\\
	\textbf{local:}		
		$r[1,\ldots,b]$; $\mathit{pos} \in \{1, \ldots, b\}$ initialized to $1$; 
		$S$ initialized to $\emptyset$; $\mathit{proposed} \in \mathbb{N}$ initialized to $1$.\\
		\setcounter{AlgoLine}{0}
		\begin{algorithm}[H]\caption{getName()}
		
		\Repeat{$(|S|+1 \geq b) \vee (r[1]=r[2]=\dots=r[b]=(S,\mathit{id}, \mathit{proposed}))$}{
			$R[\mathit{pos}].\text{write}(S, \mathit{id}, \mathit{proposed})$ \label{line:obs:write}\\
      $r[1,\ldots,b] := \Reg.\text{scan()}$ \label{line:obs:scan}\\ 
      $S := \Update(S,r[1,\ldots,b])$ \label{line:obs:update}\\
      $\mathit{proposed} = \min\{i \in  \mathbb{N}~|~i \notin \Names{S}\}$\label{line:obs:name} \\
      \uIf(\label{line:obs:if}){$\exists i$, s.t. $(r[i].\mathit{writer}=\mathit{id}) \wedge (r[i]\neq (S,\mathit{id},\mathit{proposed}))$}{
        $\mathit{pos} := \max\{i~|~(r[i].\mathit{writer}=\mathit{id}) \wedge (r[i]\neq (S,\mathit{id},\mathit{proposed}))\}$\label{line:obs:max}
      }
      \ElseIf(\label{line:obs:elseif}){$\exists j$, s.t. $r[j] \neq (S,\mathit{id},\mathit{proposed})$}
				{$\mathit{pos} := j$}	\label{line:obs:endif}
		
		 }\label{line:obs:repeat:end}

		\uIf(\label{line:obs:fewIf}){$|S|+1 \leq b - 1$}{
      \Return{$\mathit{proposed}$}\label{line:obs:ret1}\; 
    }
    \Else{
      \Return{$b - 1 + \mathit{id}$} \label{line:obs:ret2} \\
    }
		\end{algorithm}
		 
		\begin{algorithm}[H]\caption{Update()}
		$S_{\mathit{new}}=\emptyset$ \\
   \For(\label{line:obs:writer:begin}){\KwSty{all} $w \in \{r[i].\mathit{writer}~|~1 \leq i \leq b\}\setminus \{\mathit{id}, \bot\}$}{
      Let $j\in \{1,\dots,b \}$ such that $r[j].\mathit{writer} = w$ \\
      $\mathit{name}_w := r[j].\mathit{proposal}$ \\
      $S_{\mathit{new}} := S_{\mathit{new}} \cup \{(w, \mathit{name}_{w})\}$ \label{line:obs:writer:add}\\
    }\label{line:obs:writer:end}
    \For(\label{line:obs:other:begin}){$\forall p \in \Procs{\bigcup_{i=1}^{b}{r[i].\mathit{set}}}\setminus (\Procs{S_{\mathit{new}}} \cup \{\mathit{id}\})$ \label{line:obs:other:for-loop}}{
      \uIf(\label{line:obs:impIf}){$\exists i,j$, $(i < j) \wedge (r[i].\mathit{writer} = r[j].\mathit{writer}) \wedge (p \in \Procs{r[j].\mathit{set}})$}{
        $\mathit{name}_p := \name{$p$}{$r[j].\mathit{set}$}$ \label{line:obs:other:high-pos}\\
      }
      \Else{
        Let $j\in \{1,\dots,b \}$ s.t. $p \in \Procs{r[j].\mathit{set}}$ \\
        $\mathit{name}_p := \name{$p$}{$r[j].\mathit{set}$}$ \label{line:obs:other:unique}\\
      }
      $S_{\mathit{new}} := S_{\mathit{new}} \cup \{(p, \mathit{name}_{p})\}$ \\
    }\label{line:obs:other:end}
    \For(\label{line:obs:old:begin}){$\forall p \in \Procs{S}\setminus (\Procs{S_{\mathit{new}}} \cup \{\mathit{id}\})$}{
      $S_{\mathit{new}} := S_{\mathit{new}} \cup \{(p, \name{$p$}{$S$})\}$ \label{line:obs:old:add}\\
    }\label{line:obs:old:end}
    \Return{$S_{\mathit{new}}$}\label{line:obs:Snew}		
		\end{algorithm}
		\caption{$(b-1)$-Bounded $k$-Adaptive Renaming}\label{fig:Ren-alg-obs}
\end{figure}

Fig.~\ref{fig:Ren-alg-obs} presents pseudo-code for an obstruction-free one-shot $(b-1)$-bounded $k$-adaptive renaming algorithm from $b$ registers assuming an atomic scan operation. In Theorem~\ref{thm:obs.final}, we show how to remove this assumption by adding an extra register. 

\subsubsection*{Algorithm Description.} 
A \emph{naming set} is a set of ordered pairs where each pair is a process id and a proposed name 
with the property that no process id occurs in more than one pair in the set. 
Let $S$ be a naming set. 
In our algorithm and the analysis we use the following notation:
\begin{compactitem}
	\item $\Procs{S}=\{x~|~(x,y) \in S\}$,
	\item $\Names{S}=\{y~|~(x,y) \in S\}$,
	\item if $(p,n) \in S$, then \name{p}{S} is $n$; otherwise it is undefined.
\end{compactitem}
The algorithm in Fig.~\ref{fig:Ren-alg-obs} employ a set  $\Reg=\{R[1], \ldots, R[b] \}$ of shared atomic registers. 
Each register $R$ stores 
an ordered triple $(\mathit{set}, \mathit{writer}, \mathit{proposal})$
where
$\mathit{set}$ is a naming set, $\mathit{writer}$ is a process id or $\bot$ (initially) and 
$\mathit{proposal}$ is a positive integer less than or equal $b - 1$. 
Each process $p$ maintains a naming set $S_p$
and alternates between write and scan operations until it terminates with a name for itself. 
Each scan returns a view, which is an atomic snapshot of the content of all registers. 
Each write by $p$ writes a triple consisting of its set $S_p$, its $\mathit{id}$ $p$, 
and its proposed name $\mathit{name}_p$, to some register $R[j]$. 
Process $p$ uses its last view and its previous value of $S_p$ 
to determine the new value of $S_p$, $\mathit{name}_p$ and $j$.

Function \Update describes how $p$ constructs $S_p$ in three steps. 
In the first step (Lines~\ref{line:obs:writer:begin}-\ref{line:obs:writer:end}), 
$p$ creates a naming set based only on the $\mathit{writer}$s and $\mathit{proposal}$s of each register in its view. 
If the view contains a $\mathit{writer}$ with more than one $\mathit{proposal}$, $p$ chooses one pair arbitrarily. 
In the second step (Lines~\ref{line:obs:other:begin}-\ref{line:obs:other:end}), $p$ augments its naming set with additional pairs for processes that are not $\mathit{writer}$s in its view but occur in the union of all naming sets in its view. 
The main issue occurs when there is some process that is paired with more than one name 
from two or more naming sets in different registers. 
In this case, if there are two such registers with the same $\mathit{writer}$ then, $p$ chooses the pair which occurs in the register with bigger index. Otherwise, $p$ picks one pair arbitrarily. 
Finally (Lines~\ref{line:obs:old:begin}-\ref{line:obs:old:end}), $p$ adds any pair $(q,n_q)$ such that $q$ exists in the previous version of $S_p$ and is not yet added.
Observe that $S_p$ is a naming set and $p \notin \Procs{S_p}$.

In Line~\ref{line:obs:name} $p$ chooses its proposal for its own name,  $\mathit{name}_p$,
to be the smallest integer that is not paired with some other process in $S_p$.

Lines~\ref{line:obs:if}-\ref{line:obs:endif} describe how $p$ sets $j$.
If there is any register in $p$'s preceeding view with $\mathit{writer}$ component equal to $p$ but with content different from $(S_p,p,\mathit{name}_p)$ 
then $p$ writes to register $R[j]$ where $j$ is the biggest index amongst these registers. 
Otherwise it writes to some register whose content is different than $(S_p,p,\mathit{name}_p)$. 
Process $p$ continues until either in some scan, 
all registers contain the same information that $p$ has written or $|S_p|$ is larger than or equal $b-1$. 
In the first case $p$ returns $\mathit{name}_p$ and in the second case it returns $b + p - 1$.

\subsubsection*{ Proof of Correctness}

\paragraph{Overview of proof} Once a process $p$ terminates with name $n_p \leq b-1$, 
the pair ($\mathit{writer}$, $\mathit{proposal}$) of every register is equal to $(p, n_p)$. 
The core idea is that after $p$ terminates, every register that is overwritten with the wrong name for $p$ or no name for $p$, 
has a distinct $\mathit{writer}$ component.
Therefore, if a subsequent scan by another process, say $q$, does not include the correct name for $p$, 
the set of processes in that scan is large and $q$ terminates with a name larger than or equal $b$.
If the set of processes in the scan is not large, 
then there is some writer that is in the $\mathit{writer}$ component of at least 2 registers. 
In that case, 
we prove that for any such pair of registers with the same writer, 
the correct name for $p$ is in the register with the larger index.
In this way, the algorithm ensures that process $q$ keeps $(p, n_p)$ in its naming set, 
and discards incorrect names for $p$.

For our proof, we use the notion and terminology for register configuration, consistent configuration, $\Index{op}$ of operation $op$, interval and content of register $R$ in configuration $C$, $\val{C}{R}$, 
as defined in Section~\ref{Algorithms}. 
Let $p$ be a process that has terminated and returned name $n_p$. 
Define $\mathit{last}_p$ to be the last scan by $p$.
For any register configuration $D$ following register configuration $C_{\Index{\mathit{last}_p}}$, 
define a set of registers $\Wrong_{p}(D) =\{R \in \Reg~|~(\val{D}{R}.\mathit{writer} \neq p) \wedge ((p,n_p) \notin \val{D}{R}.\mathit{set}) \}$ 
and a set of processes $\Writer_{p}(D)=\bigcup_{R \in \Wrong_{p}(D)} \{ \val{D}{R}.\mathit{writer} \}$.

\begin{lemma}\label{important}
 Let $E$ be any execution starting in the initial configuration and ending in configuration $C$. 
If there are two integers $i$ and $j$ such that $i < j$, $\val{C}{R[i]}.\mathit{writer}=\val{C}{R[j]}.\mathit{writer}=p$ and $\val{C}{R[i]} \neq \val{C}{R[j]}$, then the last write to $R[i]$ happens before the last write to $R[j]$ and both are by the same process.
\end{lemma}

\begin{proof}
  By Line~\ref{line:obs:write}, the $\mathit{writer}$ segment of each register indicates the id of the process which writes that value. 
  Hence, $\val{C}{R[i]}$ and $\val{C}{R[j]}$ are both written by the same process $p$.
  Let $w_i$ and $w_j$ be the most recent writes to $R[i]$ and $R[j]$ preceding $C$, respectively. 
	Thus, value of $w_i$ (respectively $w_j$) is $\val{C}{R[i]}$ (respectively $\val{C}{R[j]}$). 
  By way of contradiction assume that $w_j$ happens before $w_i$. 
  Let $s_i$ be the most recent scan operation by $p$ before $w_i$. 
  Hence $s_i$ happens after $w_j$ and before $w_i$. 
	Since $w_j$ is the most recent write to $R[j]$ preceding $s_i$, $\val{C_{\Index{s_i}}}{R[j]} = \val{C}{R[j]}$. 
	Let $\widehat{S_p}$ and $\widehat{\mathit{proposed}_p}$ be the value of $S_p$ and $\mathit{proposed}_p$ at $w_i$ respectively. 
Then $(\widehat{S_p}, p, \widehat{\mathit{proposed}_p}) = \val{C}{R[i]}$. Therefore, when $p$ executed Line~\ref{line:obs:if} preceding $w_i$ and after $s_i$, $\widehat{S_p}$ and $\widehat{\mathit{proposed}_p}$ are values of $S_p$ and $\mathit{proposed}_p$ respectively. 
	Let $\widehat{r_j}$ be the value of $r[j]$ at $s_i$. Thus $\widehat{r_j} = \val{C_{\Index{s_i}}}{R[j]}$. 
	Furthermore, $\widehat{r_j}$ is the value of $r[j]$, when $p$ executes Line~\ref{line:obs:if} after $s_i$ and preceding $w_i$. 
	Therefore at the execution of Line~\ref{line:obs:if} after $s_i$ and preceding $w_i$, $(S_p, p, \mathit{proposed}_p) = (\widehat{S_p}, p, \widehat{\mathit{proposed}_p}) = \val{C}{R[i]} \neq \val{C}{R[j]} = \val{C_{\Index{s_i}}}{R[j]}  = \widehat{r_j} = r[j]$ and $r[j].\mathit{id} = \widehat{r_j}.\mathit{id} = p$. Thus, Line~\ref{line:obs:if} evaluates to true. 
	Since $j > i$, by Line~\ref{line:obs:max}, $p$ does not write into $R[i]$ before writing into $R[j]$.
	
\end{proof}

Informally,  Lemma~\ref{lem:core} says that every register that contains an incorrect name for $p$ 
after a consistent configuration containing the correct name for $p$ has a distinct $\mathit{writer}$ component.

\begin{lemma}\label{lem:core}
Consider an execution $E$ in which process $p$'s \getName{} call returns name $n_p \leq b-1$. 
Then for any register configuration $C_{e}$ where $e \geq \Index{\mathit{last}_p}$,
   
 \begin{compactenum}[i)]
 	\item $|\Wrong_{p}(C_{e})|=|\Writer_{p}(C_{e})|$;
 	\item $\forall q \in \Writer_{p}(C_{e}) $, $q$ performs a write in Interval$[\Index{\mathit{last}_p},e]$; and
	\item for any write operation $o$ by  any process $q$ during Interval$[\Index{\mathit{last}_p},e]$, let $v$ be the value of $o$. 
	If $o$ is not $q$'s first write during Interval$[\Index{\mathit{last}_p},e]$, 
	then $(p,n_p)\in v.\mathit{set}$.

 \end{compactenum}
\end{lemma}

\begin{proof}
Let $C_{\Index{\mathit{last}_p}} ,C_{\Index{\mathit{last}_p}+1},\ldots$ be the sequence of all register configurations starting at $C_{\Index{\mathit{last}_p}}$.  
We prove the lemma by induction on the indices of this sequence. 
Let $\widehat{S_p}$ be the value of $S_p$ at $\mathit{last}_p$. 
For the base case $\ell = \Index{\mathit{last}_p}$, 
since $n_p \leq b-1$, $p$ returns in Line~\ref{line:obs:ret1}. 
  Therefore the condition $r[1]=\dots=r[b]=(\widehat{S_p},p,n_p)$ held when $p$ last executed Line~\ref{line:obs:repeat:end}. 
  Hence, condition $R[1]=\dots=R[b]=(\widehat{S_p},p,n_p)$ held at $C_{\Index{\mathit{last}_p}}$. 
  Therefore the induction hypothesis $(i)$ and $(ii)$ hold for the base case $\ell=\Index{\mathit{last}_p}$ because 
  $\Wrong_{p}(C_{\Index{\mathit{last}_p}}) = \Writer_{p} (C_{\Index{\mathit{last}_p}})= \emptyset$. 
	Furthermore, since Interval$[\Index{\mathit{last}_p},\Index{\mathit{last}_p}]$ contains only one write, $(iii)$ is true for the base case. 
	
  Suppose that the lemma holds for $\ell-1 \geq \Index{\mathit{last}_p}$. 
  Let $w$ be the write that changes register configuration $C_{\ell-1}$ to $C_{\ell}$, and let $x$ be the process that performs $w$.
  Then clearly $x \neq p$, since $p$ has performed its last write before $C_{\Index{\mathit{last}_p}}$.
  Suppose $w$ writes value $(\widehat{S_x},x,n_x)$ into register $R$, and let $s$ be $x$'s scan operation that precedes $w$ if it exists. 

Suppose $(p,n_p) \in \widehat{S_x}$. Let $\Wrong_{p}(C_{\ell})=\Wrong_{p}(C_{\ell-1}) \setminus \{R\}$ and $\Writer_{p}(C_{\ell}) = \Writer_{p}(C_{\ell-1}) \setminus \{\val{C_{\ell-1}}{R}.\mathit{writer}\}$. 
If $R \in \Wrong_{p}(C_{\ell-1})$, then by definition, $\val{C_{\ell-1}}{R}.\mathit{writer} \in \Writer_{p}(C_{\ell-1})$ and if $R \notin \Wrong_{p}(C_{\ell-1})$, then by definition, $\val{C_{\ell-1}}{R}.\mathit{writer} \notin \Writer_{p}(C_{\ell-1})$. 
Since $|\Wrong_{p}(C_{\ell - 1})|=|\Writer_{p}(C_{\ell - 1})|$, $|\Wrong_{p}(C_{\ell})|=|\Writer_{p}(C_{\ell})|$. 
Therefore $(i)$ is true. Since $\Writer_{p}(C_{\ell}) \subseteq \Writer_{p}(C_{\ell-1})$, $(ii)$ holds.
Since $(p, n_p) \in \widehat{S_x}$, $(iii)$ is true. 

Now consider the case $(p,n_p) \notin \widehat{S_x}$. We first show that $C_{\Index{s}}$ precedes $C_{\Index{\mathit{last}_p}}$ in $\Gamma_E$ or $w$ is the first write by $x$. 
Suppose, for the purpose of contradiction, that $\Index{\mathit{last}_p} \leq \Index{s} \leq \ell - 1$. 
First consider the case that there exists an $i$ such that $\val{C_{\Index{s}}}{R[i]}.\mathit{writer} =p$. Since at $C_{\Index{last_p}}$ all registers contain $(\widehat{S_p}, p, n_p)$ and $p$ does not write after $last_p$, $\val{C_{\Index{s}}}{R[i]}.\mathit{proposal} =n_p$. Hence by Line~\ref{line:obs:writer:add}, $(p,n_p) \in \widehat{S_x}$.  
Otherwise suppose that in $C_{\Index{s}}$, there are at least two distinct registers whose $\mathit{writer}$ are the same process and not $p$. 
Then, choose any indices $i,j$ such that $i<j$ and $\val{C_{\Index{s}}}{R[i]}.\mathit{writer} = \val{C_{\Index{s}}}{R[j]}.\mathit{writer} = u \neq p$. 
Let $w_1$ and $w_2$ be the most recent writes to $R[i]$ and $R[j]$ preceding $C_{\Index{s}}$. Hence $w_1$ has value $\val{C_{\Index{s}}}{R[i]}$ and $w_2$ has value $\val{C_{\Index{s}}}{R[j]}$ and they both are performed by process $u$. 
Furthermore, since at $C_{\Index{last_p}}$ all registers contain $(\widehat{S_p}, p, n_p)$, $w_1$ and $w_2$ occur in Interval$[\Index{\mathit{last}_p},\Index{s}]$. 
Suppose $\val{C_{\Index{s}}}{R[i]} \neq \val{C_{\Index{s}}}{R[j]}$, then by Lemma~\ref{important}, $w_1$ precedes $w_2$ in $E$ 
and by the induction hypothesis~$(iii)$, $(p,n_p) \in \val{C_{\Index{s}}}{R[j]}.\mathit{set}$. 
Otherwise suppose $\val{C_{\Index{s}}}{R[i]} = \val{C_{\Index{s}}}{R[j]}$ then again by induction hypothesis~$(iii)$, $(p,n_p) \in \val{C_{\Index{s}}}{R[j]}.\mathit{set}$. 
In either case, when $x$ performs Line~\ref{line:obs:impIf} after $s$ and preceding $w$, this line evaluates to true. 
Hence by Line~\ref{line:obs:other:high-pos}, $(p,n_p) \in \widehat{S_x}$. 
Finally, if $\forall i, j$, $1 \leq i,j \leq b$ and $ i \neq j$, $\val{C_{\Index{s}}}{R[i]}.\mathit{writer} \neq \val{C_{\Index{s}}}{R[j]}.\mathit{writer}$, then by induction hypothesis $(i)$, $|\Wrong_{p}(C_{\Index{s}})|=|\Writer_{p}(C_{\Index{s}})| = b$.  
Therefore, by the for-loop (Lines~\ref{line:obs:writer:begin}-\ref{line:obs:writer:end}), $|\widehat{S_x}| + 1 \geq |\Writer_{p}(C_{\Index{s}})| = b$. 
Hence, by Line~\ref{line:obs:repeatCondition}, the presumed write $w$ by $x$ cannot happen. 
Thus, in all cases, we have established that if $(p,n_p) \notin \widehat{S_x}$ then  
$C_{\Index{s}}$ precedes $C_{\Index{\mathit{last}_p}}$ in $\Gamma_E$ or $w$ is the first write by $x$.

Consequently, any write by $x$ before $w$ happens before $C_{\Index{\mathit{last}_p}}$. 
On the other hand, by the induction hypothesis, 
 for all $q \in \Writer_{p}(C_{\ell-1})$, $q$ performs a write during Interval$[\Index{\mathit{last}_p},\ell-1]$ 
implying $x$ is not in $\Writer_{p}(C_{\ell-1})$. 
Thus by defining $\Wrong_{p}(C_{\ell})=\Wrong_{p}(C_{\ell-1}) \cup \{R\}$ and $\Writer_{p}(C_{\ell}) = \Writer_{p}(C_{\ell-1})\setminus \{\val{C_{\ell - 1}}{R}.\mathit{writer}\} \cup \{x\}$, the induction hypothesis $(i)$ and $(ii)$ hold for $\ell$. 
Since the most recent operation before $w$ by $x$ happens before $C_{\Index{\mathit{last}_p}}$, $x$ performs only one write operation in Interval$[\Index{\mathit{last}_p},\ell]$. Therefore, 
$(iii)$ holds for $\ell$.  

\end{proof}

\begin{lemma}\label{lem:m-1_distinct}
Let $p$ and $q$ be two distinct processes that have terminated in execution $E$ and returned names $n_p$ and $n_q$ respectively. Suppose that $C_{\Index{\mathit{last}_p}}$ precedes $C_{\Index{\mathit{last}_q}}$ in $\Gamma_E$. If $n_p, n_q \leq b-1$, then $|\Wrong_{p}(C_{\Index{\mathit{last}_q}})| = 0$.
\end{lemma}

\begin{proof}
Since $n_q \leq b-1$, $q$ returns in Line~\ref{line:obs:ret1}. 
Hence $C_{\Index{\mathit{last}_q}}$ is consistent with content $(S_q,q,n_q)$. 
Therefore, $|\Wrong_{p}(C_{\Index{\mathit{last}_q}})|\in \{0,b\}$.
By Lemma~\ref{lem:core}, $|\Wrong_{p}(C_{\Index{\mathit{last}_q}})|=|\Writer_{p}(C_{\Index{\mathit{last}_q}})|$. 
Since in $C_{\Index{\mathit{last}_q}}$, $R.\mathit{writer}=q$ for all $R\in\Reg$, $| \Writer_{p}(C_{\Index{\mathit{last}_q}})|\leq 1$, and thus $|\Wrong_{p}(C_{\Index{\mathit{last}_q}})| \leq 1$. Therefore $|\Wrong_{p}(C_{\Index{\mathit{last}_q}})|=0$.

\end{proof}

\begin{lemma}\label{lem:obs:name_distinctness}
The names returned by any two distinct processes are distinct.
\end{lemma}

\begin{proof}
For any two distinct processes $p$ and $q$, let $n_p$ and $n_q$ be the names returned by $p$ and $q$, respectively.
Let $\widehat{S_p}$ (respectively, $\widehat{S_q}$) be the value of $S_p$ (respectively, $S_q$) when $p$ (respectively, $q$) 
executes Line~\ref{line:obs:fewIf}. 
If $|\widehat{S_p}|,|\widehat{S_q}| \geq b-1$, the names returned by $p$ and $q$ in Line~\ref{line:obs:ret2} are distinct because $p \neq q$.

Consider the case $|\widehat{S_p}| \leq b-2$ and $|\widehat{S_q}| \geq b-1$. 
Process $p$ returns $n_p$ in Line~\ref{line:obs:ret1}. 
Since $|\Names{S_p}| \leq |S_p|$, by Line~\ref{line:obs:name}, 
$n_p$ must be smaller than or equal to $b-1$. 
Furthermore the name returned by $q$ in Line~\ref{line:obs:ret2} is larger than or equal $b$. 
The case $|\widehat{S_q}| \leq b-2$ and $|\widehat{S_p}| \geq b-1$ is true by symmetry. 
Consider the case $|\widehat{S_p}|,|\widehat{S_q}| \leq b-2$ implying $n_p, n_q \leq b-1$. 
Without loss of generality 
assume that $C_{\Index{\mathit{last}_p}}$ precedes $C_{\Index{\mathit{last}_q}}$ in $\Gamma_E$. 
By Lemma~\ref{lem:m-1_distinct}, 
$|\Wrong_p(C_{\Index{\mathit{last}_q}})|=0$. Since $\forall R \in \Reg$, $\val{C_{\Index{\mathit{last}_q}}}{R}.\mathit{writer} = q \neq p$, $(p, n_p) \in \val{C_{\Index{\mathit{last}_q}}}{R}.\mathit{set}$. 
Thus by Line~\ref{line:obs:other:for-loop}, $(p,n_p) \in \widehat{S_q}$. Therefore by Line~~\ref{line:obs:name}, $\mathit{proposed}_q \neq n_p$.

\end{proof}

\begin{observation}\label{obs:obv2}
Let $\widehat{S_p}$ be the value of $S_p$ created by \Update in Line~\ref{line:obs:name} following $p$'s scan operation $scan_p$ in Line~\ref{line:obs:scan}. Then $\forall q \in \Procs{\widehat{S_p}}$, $q$ performs at least one write before $scan_p$.
\end{observation}

\begin{lemma}\label{lem:obs:name-space}
Let $k$ be the number of participating processes during process $p$'s \getName{}. Then, the name returned by $p$, is in the range $\{1,\ldots,k\}$, if $k \leq b-1$ and in the range $\{1,\ldots,n+b-1\}$, if $k \geq b$. 
\end{lemma}

\begin{proof}
Let $\widehat{S_p}$ be the value of $S_p$ when $p$ executes Line~\ref{line:obs:update} for the last time and $n_p$ be the name returned by $p$.
By Observation~\ref{obs:obv2}, $\forall q \in \Procs{\widehat{S_p}}$, $q$ performs at least one write before $p$ returns. Thus, $\forall q \in \Procs{\widehat{S_p}}$, $q$ is a participating process. Hence, $|\widehat{S_p}|+1 \leq k$. 

If $k \leq  b-1$, process $p$ returns in Line~\ref{line:obs:ret1}. By definition, $|\Names{\widehat{S_p}}| \leq |\widehat{S_p}| \leq k-1$. Therefore by Line~\ref{line:obs:name}, $n_p \leq |\Names{\widehat{S_p}}| +1 \leq k$.

If $k \geq b$, then $p$ returns either in Line~\ref{line:obs:ret1} or in Line~\ref{line:obs:ret2}. Therefore the name is in the range $\{1,\ldots,b+n-1\}$. 

\end{proof}

\begin{theorem}\label{thm:obs.final}
  For any $b\geq 2$ there is an obstruction-free $(b-1)$-bounded $k$-adaptive renaming algorithm implemented from $b+1$ bounded registers such that when $k \geq b$ the returned names are in the range $\{1,\ldots,n+b-1\}$.
\end{theorem}
\begin{proof}
There is an obstruction-free implementation of $b$-component snapshot objects from $b+1$ bounded registers~\cite{GHHW13}.
  Since our algorithm in Fig.~\ref{fig:Ren-alg-obs} is deterministic we can replace the atomic scan registers with a linearizable scan. 
	By Lemma~\ref{lem:obs:name-space} and Lemma~\ref{lem:obs:name_distinctness}, the algorithm solves $(b-1)$-bounded $k$-adaptive renaming.
  Thus, it suffices to prove that the algorithm is obstruction-free.

If $p$ runs alone then the value of $S_p$ computed in Line~\ref{line:obs:update} and $\mathit{proposed}_p$ computed in Line~\ref{line:obs:name} remain the same.
Therefore after $b$ write operations all registers contain $(S_p,p, \mathit{proposed}_p)$. 
Therefore, in the $b$-th iteration of the repeat-until loop (Line~$11$) evaluates to true and $p$ stops.

\end{proof}

Let $f:\{1,\dots,n\}\rightarrow \mathbb{N}$ be a non-decreasing function where, $\forall k \in \{1,\dots, n\}$, $f(k) \geq k$ and $f(1) \leq n-1$. 
Let $d' = \min\{n, x ~|~ f(x) \geq 2n\}$. Hence, $d' \leq n$. 
Setting $b =  d'$, we have an obstruction-free one-shot $d'$-bounded $k$-adaptive renaming algorithm from $d' + 1$ registers. This implies that the algorithm returns names in the range $\{1, \ldots, k\}$ when $k \leq d'-1$, and returns names in the range $\{1,\ldots, n + d'-1 \}$ when $k \geq d'$. Note that $\forall k \in \{1,\dots,d'-1\}$, $k \leq f(k)$. Furthermore, when $k \geq d'$, $ n + d' - 1 < 2n \leq f(d')$. 
Hence, $\forall k \in \{1,\dots, n \}$, the algorithm returns names in the range $\{ 1,\ldots, f(k)\}$. 

\begin{corollary}
There is an obstruction-free one-shot $f$-adaptive renaming algorithm implemented from $\min\{n, x ~|~ f(x) \geq 2n\} + 1$ bounded registers.
\end{corollary}
\section{Observations and Open Problems}
Let $f:\{1,\dots,n\}\rightarrow \mathbb{N}$ be a non-decreasing function where, $\forall k \in \{1,\dots,n\}$, $f(k) \geq k$ and $f(1) \leq n-1$. Let $\d = \max\{x ~|~ f(x) \leq n-1\}$. 
We proved a lower bound of $\d + 1$ for non-deterministic solo-terminating long-lived $f$-adaptive renaming. Furthermore, for any integer constant $0 \leq c \leq n$, we showed a lower bound of $\lfloor \frac{2(n - c)}{c+2} \rfloor$ for non-deterministic solo-terminating one-shot $(k+c)$-adaptive renaming. This implies a tight space bound of $n$ for both one-shot and long-lived tight renaming. We also presented an obstruction-free one-shot $f$-adaptive algorithm from $\min\{n, x ~|~ f(x) \geq 2n\} + 1$ registers. 

An obvious solution for any obstruction-free long-lived or one-shot $f$-adaptive renaming is as follows. 
A set $Q \subseteq \Proc$ of $\lfloor f(1) \rfloor - 1$ processes always return names in the range $\{1, \ldots, \max (\lfloor f(1) \rfloor - 1, 1)\}$ without taking any steps. 
In any $(\Proc \backslash Q)$-solo execution, process in $\Proc \backslash Q$ using universal construction, get names in the range $\{ \lfloor f(1) \rfloor, \ldots, k + \lfloor f(1) \rfloor - 1 \}$. Universal construction for $|\Proc \backslash Q|$ processes requires $|\Proc| - |Q| = n - \lfloor f(1) \rfloor + 1$ registers. 
Observe that this is a tight upper bound for obstruction-free long-lived $(k+c)$-adaptive renaming. 
One of the most noticeable open problems is whether implementing one-shot $f$-adaptive renaming requires asymptotically less space than long-lived $f$-adaptive renaming. 

We designed a wait-free one-shot $(b-1)$-bounded $(k(k+1)/2)$-adaptive renaming algorithm from $b$ bounded registers, 
and established  that this algorithm has a polynomial step complexity. 
It appears that if we modify the \newScan{} function of our algorithm, so that each process returns when the set of all  processes know to it grows even by one, the step complexity would reduce considerably. 
However this change would require much more elaborate and challenging proofs because  
the values returned by \newScan{} would not be equivalent to values returned by a linearizable scan. 

For some systems, it seems reasonable to have the register space, as well as the name space,
adapt to the actual number of participants. 
The one-shot lower bound can also be modified to express the actual register use as a function of $k$. 
On the other hand, the one-shot algorithms in this paper require a fixed number of registers regardless of the number of participants. 
\bibliographystyle{plain}
\bibliography{ref}

\end{document}